\documentclass[journal, twoside, draftclsnofoot, 12pt, onecolumn]{IEEEtran}
 \makeatletter

\newcommand{\Rmnum}[1]{\expandafter\@slowromancap\romannumeral #1@}
\makeatother
\usepackage{graphicx,cite,epsfig,amssymb,amsmath,multirow}
\usepackage{graphicx,amsmath,amssymb,amsfonts,cite}
\usepackage{mathrsfs}
\usepackage{algorithm}
\usepackage{algorithmic}
\usepackage{float}
\usepackage{bm}
\usepackage{booktabs}
\usepackage{epstopdf}
\usepackage{ulem}
\usepackage{slashbox}
\usepackage{subfigure}
\usepackage{amsthm} 
\usepackage{color}

\newtheorem{proposition}{\bf{Proposition}}

\makeatother
\begin{document}
\bibliographystyle{ieeetr}

\title{ Hybrid Precoding Based on Non-Uniform Quantization Codebook to Reduce Feedback Overhead in Millimeter Wave MIMO Systems}

\author{Yun~Chen,~Da~Chen,~and~Tao~Jiang
\thanks{Y.~Chen,~D.~Chen,~and~T.~Jiang are with School of Electronic Information and Communications, Huazhong University of Science and Technology, Wuhan 430074, China (e-mail: chen\_yun@hust.edu.cn; chenda@hust.edu.cn; tao.jiang@ieee.org).}}%

\markboth{Submitted to IEEE Transactions on Communications ,~Vol.~, No.~, 2018}%
{Chen \MakeLowercase{\textit{et al.}}: Hybrid precoding based on NUQ codebook to Reduce Feedback Overhead in millimeter wave MIMO systems} 

\maketitle
\begin{abstract}
In this paper, we focus on the design of the hybrid analog/digital precoding in millimeter wave multiple-input multiple-output (MIMO) systems. To reduce the feedback overhead, we propose two non-uniform quantization (NUQ) codebook based hybrid precoding schemes for two main hybrid precoding implementations, i.e., the full-connected structure and the sub-connected structure. Specifically, we firstly group the angles of the arrive/departure (AOAs/AODs) of the scattering paths into several spatial lobes by exploiting the sparseness property of the millimeter wave in the angular domain, which divides the total angular domain into effective spatial lobes' coverage angles and ineffective coverage angles. Then, we map the quantization bits non-uniformly to different coverage angles and construct NUQ codebooks, where high numbers of quantization bits are employed for the effective coverage angles to quantize AoAs/AoDs and zero quantization bit is employed for ineffective coverage angles. Finally, two low-complexity hybrid analog/digital precoding schemes are proposed based on NUQ codebooks. Simulation results demonstrate that, the proposed two NUQ codebook based hybrid precoding schemes achieve near-optimal spectral efficiencies and show the superiority in reducing the feedback overhead compared with the uniform quantization (UQ) codebook based works, e.g., at least $12.5\%$ feedback overhead could be reduced for a system with $144/36$ transmitting/receiving antennas.
\end{abstract}

\begin{IEEEkeywords}
Millimeter wave communication, hybrid precoding, feedback overhead, spatial lobe, non-uniform quantization.
\end{IEEEkeywords}

\section{Introduction}
\IEEEPARstart{T}{he} millimeter wave communication could achieve high data rates and has caused widespread concern owing to the large bandwidth \cite{1Rappaport2013, 2pi2011, 3Bai2014, 4Heath2016, 5Ayach2013}. However, millimeter wave signals suffer severe path-loss due to the use of very high carrier frequency on the order of 30-300 GHz. Fortunately, the small wavelength of millimeter wave signals enables the deployment of large antenna arrays in small physical dimensions and make massive multiple-input multiple-output (MIMO) practical in wireless communications \cite{6Dai2016, 7Alkhateeb2014}. The large antenna arrays could provide sufficient antenna gains to compensate for the high path loss of millimeter wave signals. Therefore, millimeter wave MIMO has been a promising candidate for future cellular networks.

The precoding is an important technology in MIMO systems since it could be utilized to transmit multiple data streams and cancel the interferences between different data streams. In millimeter wave MIMO systems, the precoding is essential and could further improve the spectral efficiency. However, the precoding in traditional MIMO systems is typically realized in the digital domain and requires expensive radio frequency (RF) chains comparable in number to the antennas, which will greatly increase the hardware cost of the millimeter wave MIMO system equipped with large antenna arrays \cite{8Jin2016, 8Jin2012}.
To address the above issue, the hybrid analog/digital precoding was proposed, where the number of the RF chains is much less than the number of antennas \cite{9zhang2005}. In the hybrid precoding, the signals are firstly precoded by a low-dimensional digital precoder to cancel the interference and allocate power, then, precoded by a high-dimensional analog precoder to produce high antenna gains.

Hybrid precoding has two main implementation structures, i.e., full-connected structure and sub-connected structure \cite{6Han2015}. 
The full-connected structure utilizes a large number of phase shifters, where each RF chain is connected to all antennas to obtain full precoding gains \cite{5Ayach2013}. The sub-connected structure sacrifices some precoding gains, where each RF chain is only connected to a subset of antennas or a subarray to reduce the required number of phase shifters \cite{16Gao2016}. In the full-connected structure, the orthogonal matching pursuit (OMP) based hybrid precoding algorithm was the first scheme proposed for millimeter wave MIMO systems to obtain the near-optimal spectral efficiency \cite{5Ayach2013}. Inspired by \cite{5Ayach2013}, there are many papers devoted to designing hybrid precoding algorithms mainly based on the alternative minimization, matrix decomposition and iterative searching\cite{11Ni2016, 12Yu2016, 13Chen2015, 14Rusu2016, 14-15Dai2016}. In the sub-connected structure, the successive interference cancelation (SIC) based hybrid precoding scheme was the first scheme proposed to obtain the near-optimal spectral efficiency\cite{16Gao2016}. In \cite{17Park2017}, a near-optimal closed-form solution was proposed for the sub-connected structure. Though the above methods proposed in \cite{11Ni2016, 12Yu2016, 13Chen2015, 14Rusu2016, 14-15Dai2016, 16Gao2016, 17Park2017} could achieve good spectral efficiency, the feedback overhead issue was not considered.

For millimeter wave MIMO systems with large antenna arrays, the feedback overhead is very large and seriously affects the communication efficiency. Therefore, it is of great importance to design hybrid precoding schemes and facilitate the limited feedback \cite{5Ayach2013}. One of the effective solutions to the feedback problem is to quantize the analog precoding matrices. When the receiver obtains analog precoding matrices of which each column is selected from predefined quantization codebooks, it only needs to feed back the selected indexes rather than large dimensional precoding matrices to the transmitter. Most prior hybrid precoding schemes designed for limited feedback millimeter wave MIMO systems were based on the uniform quantization (UQ) codebooks that only depend on the single parameter (the angle of arrive/departure (AOA/AOD)) quantization and uniformly divide the total angular domain into $2^b$ ($b$ is the number of quantization bits) quantized angle parts \cite{18Alkhateeb2015, 19Lin2017, 20Alkhateeb2014, 15Singh2015}.
However, the sparseness property of the millimeter wave in the angular domain has not been fully utilized in existing works to reduce the feedback overhead.

In this paper, we focus on the design of hybrid precoding for both full-connected and sub-connected structures to reduce the feedback overhead in millimeter wave MIMO systems. The key idea is to construct novel non-uniform quantization (NUQ) codebooks that map the quantization bits non-uniformly for different coverage angles.
According to the measurement results of NYU WIRELESS \cite{21Samimi2014, 22Samimi2016, 23Rappaport2015}, the AOAs/AODs of the paths in the millimeter wave channel could be grouped in several separated spatial lobes (SLs), which enables us to divide the total angular domain into effective spatial lobes' coverage angles and ineffective coverage angles. Therefore, we employ high numbers of quantization bits to quantize both AoAs and AoDs for effective spatial lobes' coverage angles to obtain the high spectral efficiency and employ zero quantization bit for ineffective coverage angles to make the total feedback overhead lower. The main contributions of this paper are summarized as follows.
\begin{itemize}
  \item By utilizing the sparseness property of the millimeter wave in the angular domain, we construct novel NUQ codebooks and propose a low-complexity NUQ-based hybrid precoding scheme for the full-connected structure, named NUQ-HYP-Full. This scheme requires smaller feedback overhead than the UQ codebooks based schemes to maintain the near-optimal spectral efficiency. Moreover, we prove that narrowing the angle range that needs to be quantized is equivalent to increasing the quantization accuracy. When the quantization accuracy is not high enough, increasing the quantization accuracy means improving the spectral efficiency.
  \item We propose NUQ codebooks and a NUQ-based hybrid precodng scheme for the sub-connected structure, named NUQ-HYP-Sub. This scheme is the first one that utilizes beamsteering based quantization codebooks to design hybrid precoding matrices and greatly reduces the feedback overhead compared with the existing hybrid precoding schemes in the sub-connected structure.
  \item We observe that the required number of quantization bits in the sub-connected structure is smaller than that in the full-connected structure to obtain near-optimal spectral efficiencies (more than $99\%$ of the spectral efficiencies achieved by the corresponding optimal unconstrained precoding schemes). Besides the spectral efficiency and power consumption, this observation provides a new insight when comparing full-connected and sub-connected hybrid precoding implementations.
\end{itemize}

Simulation results demonstrate that the proposed NUQ-HYP-Full scheme outperforms the UQ based OMP scheme and achieves similar spectral efficiency as the fully digital precoding scheme (whose spectral efficiency is the upper bound for the full-connected structure). Moreover, the proposed NUQ-HYP-Sub scheme achieves similar spectral efficiency as the SIC based hybrid precoding scheme and the optimal hybrid precoding scheme for the sub-connected structure.

The rest of the paper is organized as follows. In Section II, the system model, channel model and the problem formulation are described. The non-quantization codebooks and the corresponding NUQ codebook based hybrid precoding schemes for the full-connected and sub-connected structures are demonstrated in Section III and Section IV, respectively.  Simulation results are presented in Section V. Finally, we conclude this paper in Section VI.

We use the following notations in this paper. $ a$ is a scalar, $\bf{a}$ is a vector, $\bf{A}$ is a matrix and ${\cal A}$ is a set. ${{\bf{A}}^{(i)}}$ is the $i_{th}$ column of $\bf{A}$ and ${\left\| {\bf{A}} \right\|_F}$ is the Frobenius norm of ${\bf{A}}$. ${{\bf{A}}^T},{{\bf{A}}^*},{{\bf{A}}^{ - 1}}$ denote the transpose, conjugate transpose and inverse of ${\bf{A}}$ respectively. ${\rm{diag}}({\bf{A}})$ is a vector that consists of diagonal elements of ${\bf{A}}$ and ${\rm{blkdiag}}({\bf{A}},{\bf{B}})$ is the block diagonal concatenation of ${\bf{A}}$ and ${\bf{B}}$. $[{\bf{A}}\left| {\bf{B}} \right.]$ is the horizontal concatenation. $\left| {\bf{a}} \right|$ is the modulus of $\bf{a}$. ${{\bf{I}}_N}$ denotes a
$N \times N$ identity matrix. $\cal L(\bf{a})$ denotes the length of $\bf{a}$. ${\cal C}{\cal N}({\bf{a}},{\bf{A}})$ is a complex Gaussian vector with mean ${\bf{a}}$ and covariance matrix ${\bf{A}}$. $\mathbb{E}[{\bf{A}}]$ is the expectation of ${\bf{A}}$.

\section{System Model, Channel Model and Problem Formulation}
\subsection{System Model}
\begin{figure*}[t]
\centering
\subfigure[]{
\begin{minipage}{0.45\textwidth}
\centering
\includegraphics[scale=0.6]{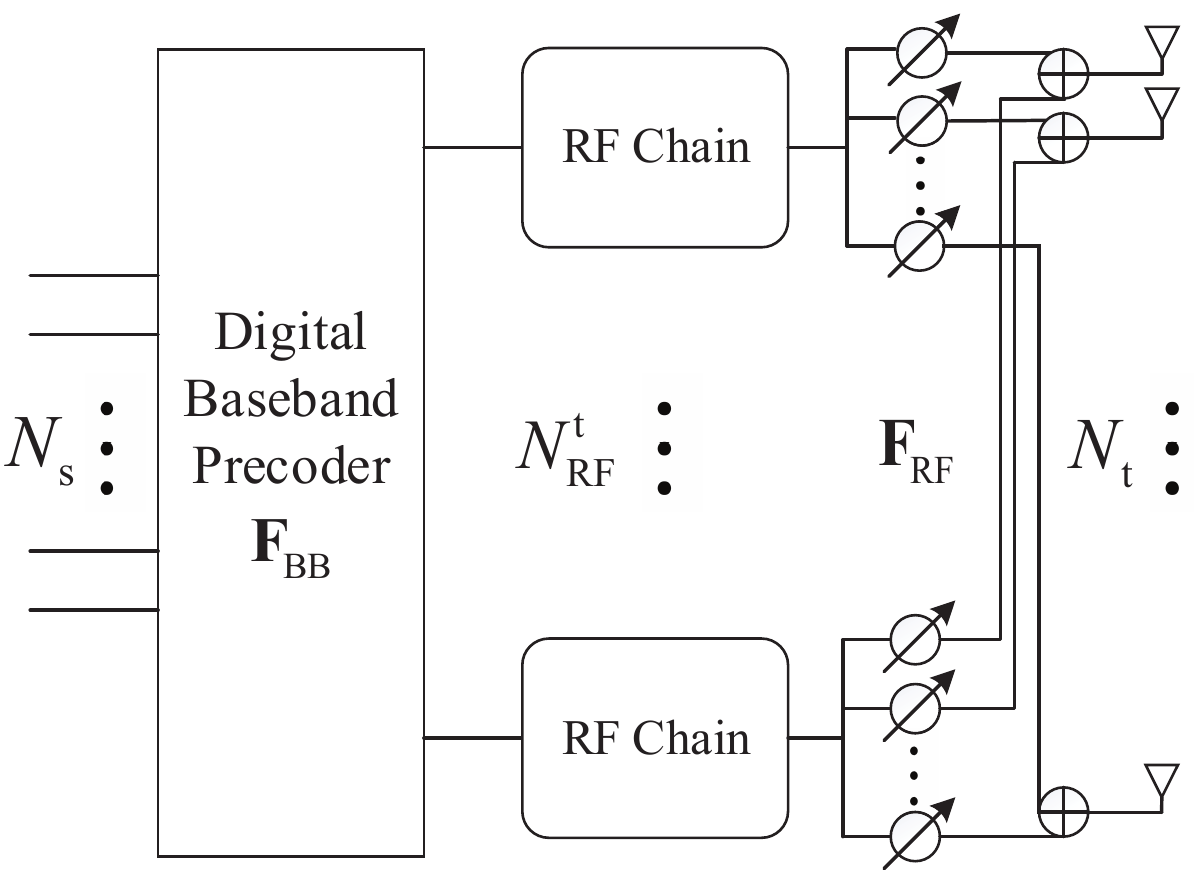}
\end{minipage}
}
\subfigure[]{
\begin{minipage}{0.45\textwidth}
\centering
\includegraphics[scale=0.6]{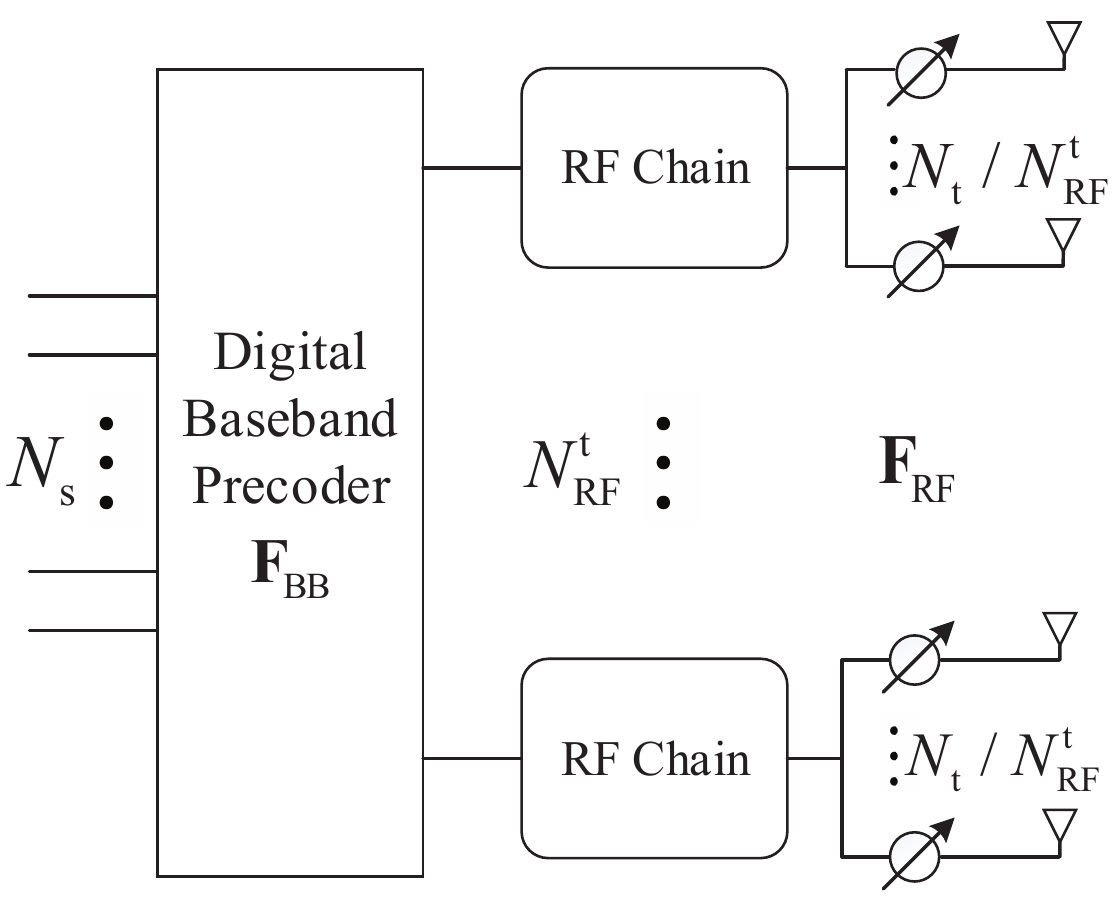}
\end{minipage}
}
\caption{Two block diagrams of the millimeter wave MIMO system models that use the hybrid precoding structure. (a) Full-connected structure, where each RF chain is connected to all antennas; (b) Sub-connected structure, where each RF chain is only connected to a subset of antennas.}
\label{fig1}
\end{figure*}

Consider two typical structures for the single user millimeter wave MIMO systems in Fig. \ref{fig1}, where Fig. \ref{fig1}(a) shows the transmitter of the full-connected structure and Fig. \ref{fig1}(b) shows the transmitter of the  sub-connected structure. In both structures, $N_{\rm{t}}/N_{\rm{r}}$ antennas and $N_{\rm{RF}}^{\rm{t}}/N_{\rm{RF}}^{\rm{r}}$ RF chains are equipped at the transmitter/receiver subject to the constrains $N_{\rm{s}}\leq N_{\rm{RF}}^{\rm{t}}\leq N_{\rm{t}}$ and $N_{\rm{s}}\leq N_{\rm{RF}}^{\rm{r}}\leq N_{\rm{r}}$, where $N_{\rm{s}}$ denotes the number of the data streams.

At the transmitter, $N_{\rm{s}}$ data streams are firstly transmitted to $N_{\rm{RF}}^{\rm{t}}$ RF chains by an $N_{\rm{RF}}^{\rm{t}}\times N_s$ baseband precoding matrix ${{\textbf{F}}_{{\rm{BB}}}}$. Then, an $N_{\rm{t}}\times N_{\rm{RF}}^{\rm{t}}$ analog precoding matrix ${{\textbf{F}}_{{\rm{RF}}}}$ transforms the data streams from RF chains to $N_{\rm{t}}$ antennas.
The discrete-time transmitted signal vector is
\begin{equation}\label{2.1}
\textbf{X}={{\textbf{F}}_{{\rm{RF}}}}{{\textbf{F}}_{{\rm{BB}}}}\textbf{s},
\end{equation}
where \textbf{s} is the $N_{\rm{s}}\times 1$ signal vector with $\mathbb{E}[{\bf{s}}{{\bf{s}}^*}] = \frac{1}{{{N_{\rm{s}}}}}{{\bf{I}}_{{N_{\rm{s}}}}}$. Since phase shifters are utilized to implement ${{\textbf{F}}_{{\rm{RF}}}}$, each entry of ${{\textbf{F}}_{{\rm{RF}}}}$ has constant amplitude constraint ${\left( {{\bf{F}}_{{\rm{RF}}}^{(i)}{\bf{F}}_{{\rm{RF}}}^{(i)*}} \right)_{l,l}} = 1/{N_{\rm{t}}}$, where ${\left(  \cdot  \right)_{l,l}}$ denotes the ${l_{th}}$ diagonal element of a matrix. In addition, the total power constrain is enforced by $\left\| {{{\bf{F}}_{{\rm{RF}}}}{{\bf{F}}_{{\rm{BB}}}}} \right\|_F^2 = {N_{\rm{s}}}$.

The narrowband block-fading channel model is adopted as shown in \cite{5Ayach2013, 20Alkhateeb2014} and the signal vector observed by the receiver is
\begin{equation}\label{2.2}
{\bf{r}} = \sqrt \rho  {\bf{H}}{{\textbf{F}}_{{\rm{RF}}}}{{\textbf{F}}_{{\rm{BB}}}}{\bf{s}} + {\bf{n}},
\end{equation}
where ${\bf{H}}$ is the ${N_{\rm{r}}} \times {N_{\rm{t}}}$ millimeter wave channel matrix, $\rho$ is the average received power, and ${\bf{n}}$ is the vector of independent and identically distributed (i.i.d.) ${\cal {CN}}(0,\sigma _{\rm{n}}^{2})$ noise.

After the combining processing at the receiver, the received signal vector is given as
\begin{equation}\label{2.3}
{\bf{y}} = \sqrt \rho  {{\bf{W}}_{{\rm{BB}}}^*}{{\bf{W}}_{{\rm{RF}}}^*}{\bf{H}}{{\textbf{F}}_{{\rm{RF}}}}{{\textbf{F}}_{{\rm{BB}}}}{\bf{s}} + {{\bf{W}}_{{\rm{BB}}}^*}{{\bf{W}}_{{\rm{RF}}}^*}{\bf{n}},
\end{equation}
where ${{\bf{W}}_{{\rm{RF}}}}$ is the ${N_{\rm{r}}} \times N_{{\rm{RF}}}^{\rm{r}}$ RF combining matrix which should satisfy ${\left( {{\bf{W}}_{{\rm{RF}}}^{(i)}{\bf{W}}_{{\rm{RF}}}^{(i)*}} \right)_{l,l}} = 1/{N_{\rm{r}}}$ and ${{\bf{W}}_{{\rm{BB}}}}$ is the $N_{{\rm{RF}}}^{\rm{r}} \times {N_{\rm{s}}}$ baseband digital combining matrix. When Gaussian symbols are transmitted through the millimeter wave channel, the achievable spectral efficiency can be written as \cite{24Goldsmith2003}
\begin{equation}\label{1.}
\begin{array}{l}
R = {\log _2}\Big (\Big | {{{\bf{I}}_{{N_{\rm{s}}}}}{\bf{ + }}\dfrac{\rho }{{{N_{\rm{s}}}}}{\bf{R}}_{\rm n}^{ - 1}{\bf{W}}_{\mathop{\rm BB}\nolimits}^*{\bf{W}}_{\mathop{\rm RF}\nolimits}^*{\bf{H}}{{\bf{F}}_{{\mathop{\rm RF}\nolimits} }}{{\bf{F}}_{{\mathop{\rm BB}\nolimits} }}}{{\bf{ \times F}}_{\mathop{\rm BB}\nolimits}^*{\bf{F}}_{\mathop{\rm RF}\nolimits}^*{{\bf{H}}^{\bf{*}}}{{\bf{W}}_{{\mathop{\rm RF}\nolimits} }}{{\bf{W}}_{{\mathop{\rm BB}\nolimits} }}} \Big |\Big ),
\end{array}
\end{equation}
where ${{\bf{R}}_{\rm{n}}} = \sigma _{\rm{n}}^2{\bf{W}}_{{\rm{BB}}}^ * {\bf{W}}_{{\rm{RF}}}^ * {{\bf{W}}_{{\rm{RF}}}}{{\bf{W}}_{{\rm{BB}}}}$ is the ${N_{\rm{s}}}\times {N_{\rm{s}}}$ noise covariance matrix after combing.

\subsection{Channel Model}
The high free-space path-loss of the millimeter wave signals leads to limited spatial scattering \cite{5Ayach2013, 26Rappaport2013}. Therefore, the geometric Saleh-Valenzuela model is usually used to represent the millimeter wave channel \cite{27Rappaport2014}, which is given by

\begin{equation}\label{2.4}
{\bf{H}}_{\rm {cl}} = \sqrt {\dfrac{{{N_{\rm{t}}}{N_{\rm{r}}}}}{{N_{\rm{cl}}{N_{\rm{ray}}}}}} \sum\limits_{m=1}^{N_{\rm{cl}}} {\sum\limits_{n=1}^{{N_{\rm{ray}}}}} {{\alpha _{m,n}}}  {{\bf{a}}_{\rm{r}}}(\theta _{m,n}^{\rm r}){{\bf{a}}_{\rm{t}}}(\theta _{m,n}^{\rm t})^*,
\end{equation}
where ${N_{\rm{cl}}}$ is the number of scattering clusters and each cluster contributes ${N_{\rm{ray}}}$ propagation paths, ${\alpha _{m,n}}$ denotes the complex gain of the $n_{th}$ path in the $m_{th}$ cluster, $\theta _{m,n}^{\rm r} \in [0,2\pi ]$ and $\theta _{m,n}^{\rm t} \in [0,2\pi ]$ are the AOA and AOD, respectively. We adopt uniform linear arrays (ULAs) at both the transmitter and the receiver. ${{\bf{a}}_{\rm{r}}}(\theta _{m,n}^{\rm r})$ and ${{\bf{a}}_{\rm{t}}}(\theta _{m,n}^{\rm t})$ are the antenna array response vectors which can be written as
\begin{equation}\label{2.5}
\begin{array}{l}
{{\bf{a}}_{\rm{t}}}(\theta _{m,n}^{\rm t}) = \dfrac{1}{{\sqrt {{N_{\rm t}}} }}\Big[1,{\kern 1pt} {e^{j(2\pi /\lambda )dsin(\theta _{m,n}^{\rm t}{\kern 1pt} )}},..., {e^{j({N_{\rm t}} - 1)(2\pi /\lambda )dsin(\theta _{m,n}^{\rm t}{\kern 1pt} )}}{\Big]^T},
\end{array}
\end{equation}
and
\begin{equation}\label{2.6}
\begin{array}{l}
{{\bf{a}}_{\rm{r}}}(\theta _{m,n}^{\rm r}) = \dfrac{1}{{\sqrt {{N_{\rm r}}} }}\Big[1,{\kern 1pt} {e^{j(2\pi /\lambda )dsin(\theta _{m,n}^{\rm r}{\kern 1pt} )}},..., {e^{j({N_{\rm r}} - 1)(2\pi /\lambda )dsin(\theta _{m,n}^{\rm r}{\kern 1pt} )}}{\Big]^T},
\end{array}
\end{equation}
respectively, where $\lambda$ is the wavelength of the signal and $d = \lambda /2$ denotes the aperture domain sample spacing.

According to the measurement results of NYU WIRELESS for the 28 GHz millimeter wave channel in Manhattan that is a typical urban environment, the  AOAs/AODs of the propagation paths could be grouped in several spatial lobes, which causes the sparseness property of the millimeter wave in the angular domain \cite{21Samimi2014, 22Samimi2016, 23Rappaport2015}. The polar plot of the millimeter wave channel measured in Manhattan at 28 GHz is shown in Fig. 2. It is observed that the angles of paths in the five dominated spatial lobes are sufficiently separated. {Moreover, the measurement results also show that the paths in several clusters may arrive in the same spatial lobe, which means the AOAs/AODs of the paths in different clusters may not be sufficiently separated.} Therefore, we reconstruct the channel from the spatial lobe's perspective to fully utilize the sparseness property of the millimeter wave in the angular domain.
Since the millimeter wave channel consists of several propagation paths, the equivalent spatial lobes millimeter wave channel could be expressed as
\begin{equation}\label{3.12}
\begin{array}{l}
{\bf{H}} = \sqrt {\dfrac{{{N_{\rm t}}{N_{\rm r}}}}{{P{Q}}}} \sum\limits_{p=1}^P {\sum\limits_{q=1}^Q {{\alpha _{p,q}}} } {{\bf{a}}_{\rm r}}({\theta _{p,q}^{\rm r})}{{\bf{a}}_{\rm t}}{(\theta _{p,q}^{\rm t})^*},
\end{array}
\end{equation}
where $P$ is the number of spatial lobes, $Q$ is the number of subpaths in one spatial lobe ($PQ={{N_{\rm{cl}}{N_{\rm{ray}}}}}$), and ${{\alpha _{p,q}}}$ is the path gains for the $q_{th}$ subpath in the $p_{th}$ spatial lobe, which obeys the Rayleigh distribution.

For convenience, the millimeter wave channel is rewritten in a more compact form as
\begin{equation}\label{2.9}
{\bf{H}}{\rm{ = }}{{\bf{A}}_{\rm{r}}}{\rm{diag}}({\boldsymbol{\alpha}}){{\bf{A}}_{\rm{t}}}^*,
\end{equation}
where ${\boldsymbol{\alpha }} = \sqrt {\frac{{{N_{\rm{t}}}{N_{\rm{r}}}}}{{PQ}}} {[{\alpha _{1,1}},{\alpha _{1,2}},...,{\alpha _{1,Q}},...,{\alpha _{{{PQ}}}}]^T}$ contains the complex gains of all paths, and the matrices
\begin{equation}\label{2.10}
\begin{split}
{{\bf{A}}_{\rm{r}}} &=\big[{{\bf{A}}_{\rm{r1}}}, {{\bf{A}}_{\rm{r2}}},..., {{\bf{A}}_{{\rm{r}}P}} \big]\\
&= \big[{{\bf{a}}_{\rm{r}}}(\theta _{1,1}^{\rm{r}}),{{\bf{a}}_{\rm{r}}}(\theta _{1,2}^{\rm{r}}),...,{{\bf{a}}_{\rm{r}}}(\theta _{1,Q}^{\rm{r}}),...,{{\bf{a}}_{\rm{r}}}(\theta _{P,Q}^{\rm{r}})\big]
\end{split}
\end{equation}
and
\begin{equation}\label{2.11}
\begin{split}
{{\bf{A}}_{\rm{t}}}&=\big[{{\bf{A}}_{\rm{t1}}}, {{\bf{A}}_{\rm{t2}}},..., {{\bf{A}}_{{\rm{t}}P}} \big]\\
&= \big[{{\bf{a}}_{\rm{t}}}(\theta _{1,1}^{\rm{t}}),{{\bf{a}}_{\rm{t}}}(\theta _{1,2}^{\rm{t}}),...,{{\bf{a}}_{\rm{t}}}(\theta _{1,Q}^{\rm{t}}),...,{{\bf{a}}_{\rm{t}}}(\theta _{P,Q}^{\rm{t}})\big]
\end{split}
\end{equation}
contain the array response vectors. Inspired by the diagonal form of (9), we find that the total number of paths is the upper bound of the rank of the millimeter wave channel matrix, which means the number of data streams $N_{\rm{s}}$ should satisfy $N_{\rm{s}}\leq PQ$ to maintain the good system performance.

\begin{figure}[t]
\centering
\includegraphics[scale=0.15]{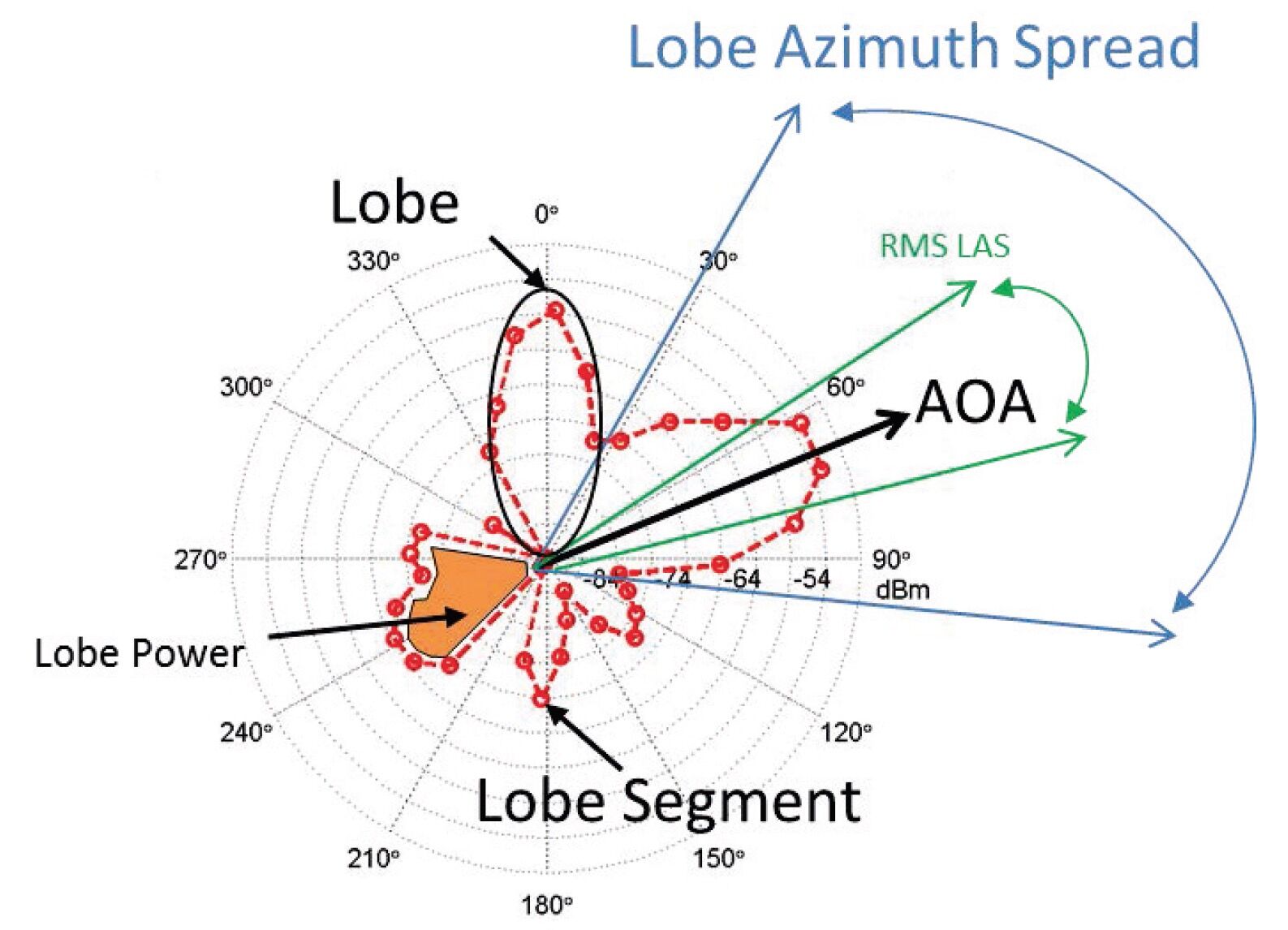}
\caption{The polar plot of the millimeter wave channel measured in Manhattan at 28 GHz
\cite{21Samimi2014}.}
\label{fig2}
\end{figure}

\subsection{Problem Formulation}
The target of designing the hybrid precoding matrices (${\bf{F}}_{{\rm{RF}}}, {\bf{F}}_{{\rm{BB}}}, {\bf{W}}_{{\rm{RF}}}, {\bf{W}}_{{\rm{BB}}}$) is to maximize the spectral efficiency. Therefore, the optimization problem could be expressed as
\begin{equation}\label{2.13}
\begin{array}{l}
({\bf{F}}_{{\rm{RF}}}^{{\rm{opt}}}, {\bf{F}}_{{\rm{BB}}}^{{\rm{opt}}}, {\bf{W}}_{{\rm{RF}}}^{{\rm{opt}}}, {\bf{W}}_{{\rm{BB}}}^{{\rm{opt}}}) = \mathop {{\rm{arg}}{\kern 1pt} {\kern 1pt} {\rm{max}}} R \\
 {\kern 83pt} {\rm{s}}{\rm{.t}}{\rm{.}}{\kern 7pt}  {{\bf{F}}_{{\rm{RF}}}} \in {{\cal F}_{{\rm{RF}}}},\\
 {\kern 102pt} {{\bf{W}}_{{\rm{RF}}}} \in {{\cal W}_{{\rm{RF}}}},\\
 {\kern 97pt} \left\| {{{\bf{F}}_{{\rm{RF}}}}{{\bf{F}}_{{\rm{BB}}}}} \right\|_{_F}^2 = {N_{\rm s}},
\end{array}
\end{equation}
where ${{\cal F}_{{\rm{RF}}}}$ and ${{\cal W}_{{\rm{RF}}}}$ are the sets of the feasible analog precoders and combiners induced by the constant amplitude constraint, respectively. Directly optimizing the problem (12) is very non-trivial due to the constant amplitude constraint on the analog precoding matrices. Based on the mathematical derivations in \cite{5Ayach2013}, the design of the precoding matrices and combining matrices are firstly decoupled, which indicates that we could focus on the design of precoding matrices $({\bf{F}}_{{\rm{RF}}}^{{\rm{opt}}}, {\bf{F}}_{{\rm{BB}}}^{{\rm{opt}}})$. The combining matrices $({\bf{W}}_{{\rm{RF}}}^{{\rm{opt}}}, {\bf{W}}_{{\rm{BB}}}^{{\rm{opt}}})$ could be obtained in the similar way except that there is no extra power constraint \cite{12Yu2016}. Then, an equivalent sparse reconstruction problem is formulated, which is aimed to minimize the Euclidean distance between the product of the analog and digital precoding matrices and the optimal unconstrained precoding matrix. 
The sparse reconstruction problem for the transmitter is given by
\begin{equation}\label{2}
\begin{array}{l}
({\bf{F}}_{{\rm{RF}}}^{{\rm{opt}}},{\bf{F}}_{{\rm{BB}}}^{{\rm{opt}}}) = \mathop {{\rm{arg}}{\kern 1pt} {\kern 1pt} {\rm{min}}}\limits_{{{\bf{F}}_{{\rm{BB}}}},{{\bf{F}}_{{\rm{RF}}}}} {\left\| {{{\bf{F}}_{{\rm{opt}}}} - {{\bf{F}}_{{\rm{RF}}}}{{\bf{F}}_{{\rm{BB}}}}} \right\|_F},{\kern 1pt} {\kern 1pt} {\kern 1pt} {\kern 1pt} {\kern 1pt} {\kern 1pt} {\kern 1pt} {\kern 1pt} {\kern 1pt} {\kern 1pt} {\kern 1pt} {\kern 1pt} {\kern 1pt} {\kern 1pt} {\kern 1pt} {\kern 1pt} \\
 {\kern 83pt} {\rm{s}}{\rm{.t}}{\rm{.}}{\kern 7pt}  {{\bf{F}}_{{\rm{RF}}}} \in {{\cal F}_{{\rm{RF}}}},\\
 {\kern 97pt} \left\| {{{\bf{F}}_{{\rm{RF}}}}{{\bf{F}}_{{\rm{BB}}}}} \right\|_{_F}^2 = {N_{\rm s}},
\end{array}
\end{equation}
where ${{\bf{F}}_{{\rm{opt}}}}$ is the optimal unconstrained reference precoding matrix which could be obtained from the singular value decomposition (SVD) of ${\bf{H}}$.

\section{Non-Uniform Quantization Codebook Based Hybrid Precoding for the Full-connected Structure}
For the full-connected structure, each RF chain is connected to all antennas, as shown in Fig. 1(a). In the design of the hybrid precoding scheme, similar to \cite{12Yu2016}, we decouple the design of the analog precoding matrices and digital precoding matrices. Since there is no constant amplitude constraint on digital precoding mareices, the digital precoding matrices could be obtained by simply SVD for the effective channel (which is defined in Section III.C).  Therefore, we mainly focus on the design of analog precoding matrices by using quantization codebooks, which means that the quantization codebooks are the key points to design the hybrid precoding matrices. We will firstly introduce the classical uniform quantization beamsteering codebooks. Then, non-uniform quantization codebooks are proposed by exploiting the sparseness property of the millimeter wave in the angular domain. Finally, a low-complexity NUQ codebook based hybrid precoding scheme for the full-connected structure is proposed.
\subsection{Uniform Quantization Codebooks}
Most of the prior works on the design of the hybrid precoding for limited feedback millimeter wave MIMO systems were based on the UQ beamsteering codebooks, since the UQ beamsteering codebooks, which are of relatively small size, only depend on the single parameter (the beamsteering angle) quantization and could provide good quantization performance \cite{5Ayach2013, 18Alkhateeb2015}. The UQ beamsteering codebooks for the transmitter and receiver can be respectively written as
\begin{equation}\label{3.}
{\bf{A}}_{\rm{t}}^{{\rm{quant}}} = \Big[{\bf a}_{\rm{t}}^{{\rm{quant}}}({\theta _1}),{\bf a}_{\rm{t}}^{{\rm{quant}}}({\theta _2}),...,{\bf a}_{\rm{t}}^{{\rm{quant}}}({\theta _{{2^b}}})\Big],
\end{equation}
and
\begin{equation}\label{3.}
{\bf{A}}_{\rm{r}}^{{\rm{quant}}} = \Big[{\bf a}_{\rm{r}}^{{\rm{quant}}}({\theta _1}),{\bf a}_{\rm{r}}^{{\rm{quant}}}({\theta _2}),...,{\bf a}_{\rm{r}}^{{\rm{quant}}}({\theta _{{2^b}}})\Big],
\end{equation}
where the entries of ${\bf{A}}_{\rm{t}}^{{\rm{quant}}}$ and ${\bf{A}}_{\rm{r}}^{{\rm{quant}}}$ are
\begin{equation}\label{3.}
\small
{\bf a}_{\rm{t}}^{{\rm{quant}}}({\theta _i}) = \dfrac{1}{{\sqrt {{N_{\rm t}}} }}\Big[1,{e^{j\pi {\rm sin}(\frac{{2\pi (i-1)}}{{{2^b}}})}},...,{e^{j({N_{\rm t}} - 1)\pi {\rm sin}(\frac{{2\pi (i-1)}}{{{2^b}}})}} \Big]^T,
\end{equation}
and
\begin{equation}\label{3.}
\small
{\bf a}_{\rm{r}}^{{\rm{quant}}}({\theta _i}) = \dfrac{1}{{\sqrt {{N_{\rm r}}} }}\Big[1,{e^{j\pi {\rm sin}(\frac{{2\pi (i-1)}}{{{2^b}}})}},...,{e^{j({N_{\rm r}} - 1)\pi {\rm sin}(\frac{{2\pi (i-1)}}{{{2^b}}})}} \Big]^T,
\end{equation}
respectively. Obviously, the larger the number of quantization bits $b$ is, the higher the quantization accuracy of the UQ beamsteering codebooks is.

\subsection{Non-Uniform Quantization Codebooks for the Full-connected Structure}
\begin{figure}[t]
\centering
\subfigure[]{
\begin{minipage}{0.45\textwidth}
\centering
\includegraphics[scale=0.58]{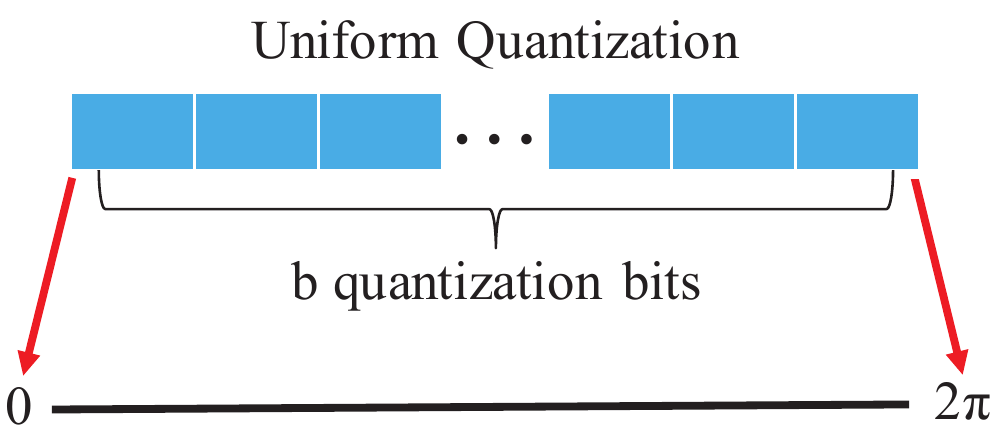}
\end{minipage}
}

\subfigure[]{
\begin{minipage}{0.45\textwidth}
\centering
\includegraphics[scale=0.58]{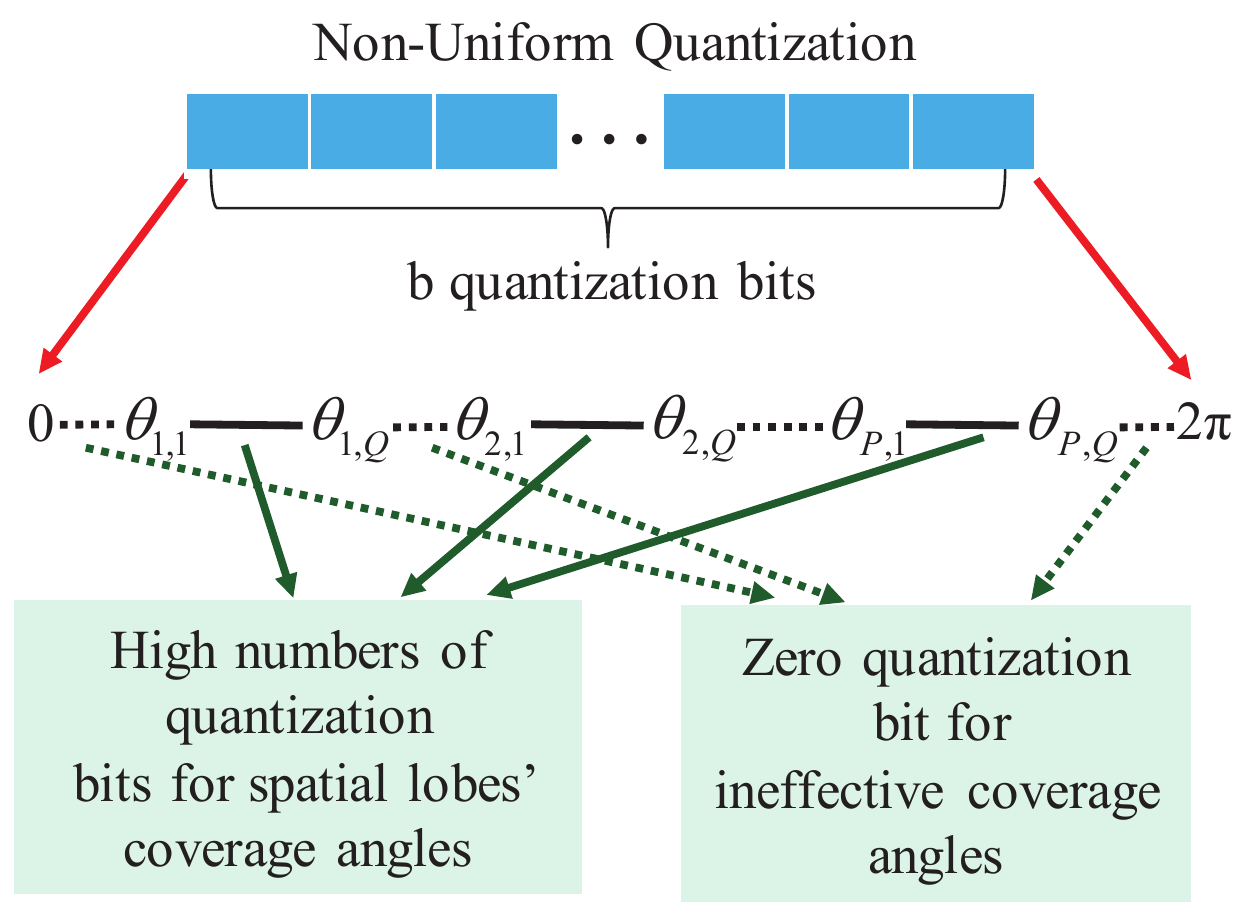}
\end{minipage}
}
\caption{Two quantization schemes. (a) Uniform-quantization scheme, where b quantization bits are uniformly mapped to the total angular domain; (b) Non-uniform quantization scheme, where b quantization bits are only mapped to spatial lobes' coverage angles.}
\label{fig2}
\end{figure}
In this subsection, we present the non-uniform quantization codebooks for the full-connected structure. Since the angles of paths in different spatial lobes are sufficiently separated, we divide the total angular domain into effective spatial lobes' coverage angles and ineffective coverage angles. Therefore, we could employ high numbers of quantization bits to quantize effective spatial lobes' coverage angles to obtain good performance and employ zero quantization bit to quantize ineffective coverage angles to maintain the average number of quantization bits unchanged.
That is we only map the quantization bits to effective spatial lobes' coverage angles.
In the meantime, without loss of generality, we assume that the transmitter and receiver have the same spatial lobes distribution, i.e., the mean angles of spatial lobes (${\widetilde{\theta}}_i, i=1,2,...,P$) are evenly distributed within $[0,2\pi]$, and the angles of the subpaths in each spatial lobe are randomly distributed with a finite angle spread ($\omega_i$) \cite{21Samimi2014, 22Samimi2016}.
Fig. 3 shows the details of the non-uniform quantization, where $\theta_{i,q}$ is the angle of the $q_{th}$ subpaths in the $i_{th}$ spatial lobe and $b$ quantization bits are non-uniformly mapped to different coverage angles.

We define the beam coverage of the $i_{th}$ spatial lobe as
\begin{equation}\label{3.6}
\begin{split}
{\bf{\cal{CV}}}({SL}_{{i}}) &= \mathop  \bigcup \limits_{{{j}} = 1,2,...,Q} {\cal{CV}}\big(a(\theta _{{{i,j}}})\big)\\
&=[\theta_{i,1}, \theta_{i,Q}]\\
&=\Big[{\widetilde{\theta}}_i-\omega_i/2, {\widetilde{\theta}}_i+\omega_i/2\Big], {{i}}=1,2,...,{{P}},
\end{split}
\end{equation}
where ${\cal{CV}}\big(a(\theta _{{{i,j}}})\big)$ is the beam coverage of the steering vector for the $j_{th}$ subpath in the $i_{th}$ spatial lobe.

The quantized angles coverage for the $i_{th}$ spatial lobe is designed as
\begin{equation}\label{3.7}
\small
{{\cal{CV}}_{i}^{{\rm{quant}}}}=[{\widetilde{\theta}}_i-{\theta^{{\rm range}}_i}/2, {\widetilde{\theta}}_i+{\theta^{{\rm range}}_i}/2],
\end{equation}
which should satisfy
\begin{equation}\label{3.8}
\omega_i\leq {\theta^{{\rm range}}_i},
\end{equation}
\begin{equation}\label{3.10}
\mathop  \bigcup \limits_{{{i}} = 1,2,...,P} {{\cal{CV}}_{ i}^{{\rm{quant}}}}\subseteq [0,2\pi],
\end{equation}
and
\begin{equation}\label{3.11}
\mathop  \bigcap \limits_{{{i}} = 1,2,...,P} {{\cal{CV}}_{ i}^{{\rm{quant}}}}=\varnothing,
\end{equation}
where ${\theta^{{\rm range}}_i}$ is the quantized angle range for the $i_{th}$ spatial lobe. (20) guarantees the codebooks are able to quantize the actual angles of all subpaths in each spatial lobe. (\ref{3.10}) is an obvious constraint. (\ref{3.11}) comes from the sparseness property of the millimeter wave in the angular domain that the angles of paths in different spatial lobes are sufficiently separated.

Denoting by $b$ the number of quantization bits, the quantized accuracy of the angle is defined as
\begin{equation}\label{3.}
\Delta_\theta=2\pi/(2^b),
\end{equation}
and the vector of quantized angles for the $i_{th}$ spatial lobe could be obtained as
\begin{equation}\label{3.12}
{\Theta}_{ i}^{{\rm{quant}}}=\Big[{\widetilde{\theta}}_i-\frac{\theta^{{\rm range}}_i}{2}, {\widetilde{\theta}}_i-\frac{\theta^{{\rm range}}_i}{2}+\Delta_\theta, ..., {\widetilde{\theta}}_i+\frac{\theta^{{\rm range}}_i}{2}\Big].
\end{equation}
For the full-connected structure employed with $N_{\rm{t}}\times N_{\rm{r}}$ antennas, the antenna indexes for the transmitter and receiver have the form
\begin{equation}\label{3.13}
{\bf A}_{\rm{index}}^{\rm{t}}=[0, 1, ..., (N_{\rm{t}}-1)],
\end{equation}
and
\begin{equation}\label{3.14}
{\bf A}_{\rm{index}}^{\rm{r}}=[0, 1, ..., (N_{\rm{r}}-1)],
\end{equation}
respectively. Therefore, the non-uniform quantization codebooks for the $i_{th}$ spatial lobe can be written as
\begin{equation}\label{3.15}
{{\bf{A}}_{{\rm{t}}i}^{{\rm{quant}}}}(:,m)=\sqrt {1/N_{\rm{t}}}*e^{(j\pi {\bf A}_{\rm{index}}^{\rm{t}}{\rm sin}(\Theta_{i}^{{\rm{quant}}}(m)))},
\end{equation}
and
\begin{equation}\label{3.16}
{{\bf{A}}_{{\rm{r}}i}^{{\rm{quant}}}}(:,m)=\sqrt {1/N_{\rm{r}}}*e^{(j\pi {\bf A}_{\rm{index}}^{\rm{r}}{\rm sin}(\Theta_{i}^{{\rm{quant}}}(m)))},
\end{equation}
where $m=1,2,...,{\cal L}(\Theta_{i}^{{\rm{quant}}})$. The non-uniform quantization codebooks for the full-connected structure are summarized in {\bf{Algorithm} \ref{alg:1}}.
\begin{algorithm}[t]
    \renewcommand{\algorithmicrequire}{\textbf{Input:}}
	\renewcommand{\algorithmicensure}{\textbf{Output:}}
	\caption{Non-Uniform Quantization (NUQ) Codebooks for the Full-connected Structure}
    \label{alg:1}
	\begin{algorithmic}[1]
        \REQUIRE 
        $ N_{\rm{t}}$, $N_{\rm{r}}$, $b$
        \ENSURE ${{\bf{A}}_{\rm t}^{{\rm{quant}}}}$, ${{\bf{A}}_{\rm r}^{{\rm{quant}}}}$
        \STATE ${\bf A}_{\rm{index}}^{\rm{t}}=[0, 1, ..., (N_{\rm{t}}-1)$], ${\bf A}_{\rm{index}}^{\rm{r}}=[0, 1, ..., (N_{\rm{r}}-1)]$
        \FOR {$i\leq P$}
        \STATE ${{\cal{CV}}_{i}^{{\rm{quant}}}}=[{\widetilde{\theta}}_i-{\theta^{{\rm range}}_i}/2, {\widetilde{\theta}}_i+{\theta^{{\rm range}}_i}/2]$
        \STATE $\Delta_\theta=2\pi/(2^b)$
        \STATE ${\Theta}_{ i}^{{\rm{quant}}}=\Big[{\widetilde{\theta}}_i-\frac{\theta^{{\rm range}}_i}{2}, {\widetilde{\theta}}_i-\frac{\theta^{{\rm range}}_i}{2}+\Delta_\theta, ..., {\widetilde{\theta}}_i+\frac{\theta^{{\rm range}}_i}{2}\Big]$
        \FOR {$m\leq {\cal L}(\Theta_{i}^{{\rm{quant}}})$}
        \STATE ${{\bf{A}}_{{\rm{t}}i}^{{\rm{quant}}}}(:,m)=\sqrt {1/N_{\rm{t}}}*e^{(j\pi {\bf A}_{\rm{index}}^{\rm{t}}{\rm sin}(\Theta_{i}^{{\rm{quant}}}(m)))}$
        \STATE${{\bf{A}}_{{\rm{r}}i}^{{\rm{quant}}}}(:,m)=\sqrt {1/N_{\rm{r}}}*e^{(j\pi {\bf A}_{\rm{index}}^{\rm{r}}{\rm sin}(\Theta_{i}^{{\rm{quant}}}(m)))} $
        \ENDFOR
        \ENDFOR
        \STATE ${{\bf{A}}_{\rm t}^{{\rm{quant}}}}=[{{\bf{A}}_{\rm t1}^{{\rm{quant}}}}, {{\bf{A}}_{\rm t2}^{{\rm{quant}}}},...,{{\bf{A}}_{{\rm t}P}^{{\rm{quant}}}}]$
        \STATE ${{\bf{A}}_{\rm r}^{{\rm{quant}}}}=[{{\bf{A}}_{\rm r1}^{{\rm{quant}}}}, {{\bf{A}}_{\rm r2}^{{\rm{quant}}}},...,{{\bf{A}}_{{\rm r}P}^{{\rm{quant}}}}]$
    \end{algorithmic}
\end{algorithm}

 {\emph{Remark 1:}} In conventional MIMO systems where channels are relatively rich, many quantization beamforming codebooks are designed to satisfy some particular properties, e.g., the Grassmannian codebooks in which the property of maximizing the minimum distance between the codebook vectors is adopted \cite{30Love2003}. However, the codebooks designed for traditional MIMO systems may not very suitable for millimeter wave MIMO systems where the channel has limited spatial scattering and large antenna arrays are employed at both the BSs and the MSs, since they are relatively complicated. Motivated by the good performance of the hybrid precoding schemes which are relied on RF beamsteering vectors, the proposed NUQ codebooks are still based on the classic beamsteering codebooks. It is worthy to point out that the idea of NUQ is general for any codebooks. For future work, it is of interest to evaluate the performance of other NUQ based beamforming codebooks for millimeter wave MIMO systems.
\begin{proposition}
Narrowing the angle range which needs to be quantized is equivalent to increasing the quantization accuracy when the number of quantization bits $b$ is fixed. In particular, the quantization accuracy could be increased by 1 bit when the total quantized angle range is narrowed by half.
\end{proposition}
\begin{proof}
For the millimeter wave channel with $P$ spatial lobes, the total effective beam coverage is
\begin{equation}\label{}
\widetilde{{\cal{CV}}}=\bigcup \limits_{{{i}} = 1,2,...,P} {\bf{\cal{CV}}}({SL}_{{i}}).
\end{equation}
According to (\ref{3.8})-(\ref{3.11}), we have
\begin{equation}\label{}
\widetilde{{\cal{CV}}}\subseteq\mathop  \bigcup \limits_{{{i}} = 1,2,...,P} {{\cal{CV}}_{ i}^{{\rm{quant}}}},
\end{equation}
and
\begin{equation}\label{3.13}
\sum\limits_{i = 1}^P{\theta^{{\rm range}}_i}\leq 2\pi.
\end{equation}
Denoting by $\widehat{b}$ the equivalent quantization bits, we have
\begin{equation}\label{}
2^b=\frac{\sum\limits_{i = 1}^P{\theta^{{\rm range}}_i}}{\Delta_\theta}=\frac{\sum\limits_{i = 1}^P{\theta^{{\rm range}}_i}}{2\pi/2^{\widehat{b}}},
\end{equation}
therefore, we could obtain
\begin{equation}\label{}
{\widehat{b}}=b-\log_2 {\frac{\sum\limits_{i = 1}^P{\theta^{{\rm range}}_i}}{2\pi}}.
\end{equation}
According to (\ref{3.13}), $\log_2 {\frac{\sum\limits_{i = 1}^P{\theta^{{\rm range}}_i}}{2\pi}}\leq 0$, therefore, we have $\widehat{b}\geq b$. Particularly, when $\sum\limits_{i = 1}^P{\theta^{{\rm range}}_i}= \pi$, which means the total quantized angle range is narrowed by half, we could easily obtain that
\begin{equation}\label{}
{\widehat{b}}=b-\log_2 {\frac{\pi}{2\pi}}=b+1,
\end{equation}
\end{proof}

As we know, when $b$ is not large enough, increasing the quantization accuracy means improving the spectral efficiency. Proposition 1 indicates that, when $b$ is fixed, the spectral efficiency could be indirectly improved by narrowing the total angle range ${\sum\limits_{i = 1}^P{{\theta^{{\rm range}}_i}}}$ that needs to be quantized.

\subsection{NUQ-Based Hybrid Precoding for the Full-connected Structure}
\begin{algorithm}[t]
    \renewcommand{\algorithmicrequire}{\textbf{Input:}}
	\renewcommand{\algorithmicensure}{\textbf{Output:}}
	\caption{NUQ Codebook Based Hybrid Precoding for the Full-connected Structure, i.e. NUQ-HYP-Full}
    \label{alg:2}
	\begin{algorithmic}[1]
        \REQUIRE 
        $ {{\bf{A}}_{\rm{t}}}$, ${{\bf{A}}_{\rm{r}}}$, ${{\bf{A}}_{\rm t}^{{\rm{quant}}}}$, ${{\bf{A}}_{\rm r}^{{\rm{quant}}}}$
		\ENSURE ${{\bf F}_{{\rm{RF}}}}$, ${{\bf F}_{{\rm{BB}}}},{{\bf W}_{{\rm{RF}}}}$, ${{\bf W}_{{\rm{BB}}}}$
		\STATE {\bf {First stage}}: Analog precoding matrices design
        \FOR {$i\leq P$}
        \STATE ${\bf{F}}_{{\rm{res}}}={{\bf{A}}_{{\rm{t}}i}}$
        \STATE ${\bf{\Psi}}  = {\bf{A}}_{{\rm{t}}i}^{{\rm{quant*}}}{{\bf{F}}_{{\rm{res}}}}$
		\FOR {$j\leq Q$}
        \STATE $k = {\rm{argma}}{{\rm{x}}_{l = 1,...,{\rm{Q}}}}{\rm{diag}}({\bf{\Psi }}{{\bf{\Psi }}^*})$
        \STATE ${{\bf{F}}_{{{\rm{RF}}_i}}} = \Big[{{\bf{F}}_{{{\rm{RF}}_i}}}\left| {{\bf{A}}_{{{\rm{t}}_i}}^{{\rm{quant}}}}(:, k) \right.\Big]$
        \STATE ${\rm{diag}}({\bf{\Psi }}{{\bf{\Psi }}^*})(k)=0$
		\ENDFOR
        \ENDFOR
        \STATE${{\bf{F}}_{{\rm{RF}}}} = [{{\bf{F}}_{{\rm{RF_1}}}},{{\bf{F}}_{{\rm{RF_2}}}},...,{{\bf{F}}_{{{\rm{BB}}_P}}}]{\kern 1pt} {\kern 1pt}$
        \STATE We could obtain ${{\bf{W}}_{{\rm{RF}}}}$ in the similar way\\
        \STATE ${{\bf{W}}_{{\rm{RF}}}} = [{{\bf{W}}_{{\rm{RF_1}}}},{{\bf{W}}_{{\rm{RF_2}}}},...,{{\bf{W}}_{{{\rm{BB}}_P}}}]$
        \STATE Compute the effective channels ${{\bf{H}}_{{{\rm{eq}}_i}}} = {\bf{W}}_{{{\rm{RF}}_i}}^{\rm{*}}{\bf{H}}{{\bf{F}}_{{{\rm{RF}}_i}}},i=1,2,...,P$
		\STATE {\bf {Second stage}}: Digital precoding matrices design
        \STATE Compute the SVD of each effective channel ${{\bf{H}}_{{{\rm{eq}}_i}}}$, ${\bf{W}}_{{{\rm{RF}}_i}}^{\rm{*}}{\bf{H}}{{\bf{F}}_{{{\rm{RF}}_i}}} = {{\bf{U}}_{{{\rm{eq}}_i}}}{{\bf{\Sigma}}_{{{\rm{eq}}_i}}}{{\bf{V}}_{{{\rm{eq}}_i}}^*}$
        \STATE ${{\bf{F}}_{{{\rm{BB}}_i}}} = {\bf{V}}_{{{\rm{eq}}_i}}$,${{\bf{W}}_{{{\rm{BB}}_i}}} = {{\bf{U}}_{{{\rm{eq}}_i}}}$
        \STATE ${{\bf{F}}_{{\rm{BB}}}} = {\rm{blkdiag}}({{\bf{F}}_{{\rm{BB_1}}}},{{\bf{F}}_{{\rm{BB_2}}}},...,{{\bf{F}}_{{{\rm{BB}}_P}}})$, ${{\bf{W}}_{{\rm{BB}}}} = {\rm{blkdiag}}({{\bf{W}}_{{\rm{BB_1}}}},{{\bf{W}}_{{\rm{BB_2}}}},...,{{\bf{W}}_{{{\rm{BB}}_P}}})$
        \STATE Normalize the digital precoding matrix at the transmitter ${{\bf{F}}_{{\rm{BB}}}} = \sqrt {{N_s}} \frac{{{{\bf{F}}_{{\rm{BB}}}}}}{{{{\left\| {{{\bf{F}}_{{\rm{RF}}}}{{\bf{F}}_{{\rm{BB}}}}} \right\|}_F}}}$
    \end{algorithmic}
\end{algorithm}

We are committed to constructing a low-complexity hybrid precoding solution (i.e., NUQ-HYP-Full) which is based on the following two operations.
\begin{itemize}
  \item The hybrid precoding matrix for each spatial lobe is designed one by one, since the AOAs/AODs of paths in different spatial lobes are sufficiently separated and thus the paths in different spatial lobes could be considered approximately orthogonal.
  \item Considering that the right and left singular matrices of $\bf H$ converge in chordal distance to the antenna response matrices, when the number of paths is much smaller than number of the antennas\cite{29Ayach2012}. We set the antenna response matrices $A_{{\rm t}i}$ and $A_{{\rm r}i}$ rather than the optimal unconstrained precoding matrices as the reference precoding matrices for the $i_{th}$ spatial lobe, since $Q\ll min(N_{\rm t}, N_{\rm r})$.
\end{itemize}
Once we obtain the NUQ codebooks or the NUQ candidate matrices, the analog precoding matrices could be obtained by searching all entries of the candidate matrices to find the vectors which are respectively closest to each entry of the reference precoding matrices in the $l_2$ norm sense.

The analog precoding matrix design problem for the $i_{th}$ spatial lobe at the transmitter could be formulated as
\begin{equation}\label{3.c1}
\begin{array}{l}
{\bf{F}}_{{\rm{RF}}_i}^{{\rm{opt}}} = \mathop {{\rm{arg}}{\kern 1pt} {\kern 1pt} {\rm{min}}} {\left\| {\bf{F}}_{{\rm{RF}}_i} - {{\bf{A}}_{{\rm t}i}} \right\|_F},{\kern 1pt} {\kern 1pt} {\kern 1pt} {\kern 1pt} {\kern 1pt} {\kern 1pt} {\kern 1pt} {\kern 1pt} {\kern 1pt} {\kern 1pt} {\kern 1pt} {\kern 1pt} {\kern 1pt} {\kern 1pt} {\kern 1pt} {\kern 1pt} \\
 {\kern 54pt} {\rm{s}}{\rm{.t}}{\rm{.}}{\kern 7pt}  {\bf{F}}_{{\rm{RF}}_i} \in {{\bf{A}}_{{\rm t}i}^{{\rm{quant}}}} ,
\end{array}
\end{equation}
which could be equivalently solved by finding the $Q$ vectors in ${{\bf{A}}_{{\rm t}i}^{{\rm{quant}}}}$ along which the reference precoding matrix has the largest $Q$ projections. The combining matrix ${{\bf{W}}_{{{\rm{RF}}_i}}}$ for the $i_{th}$ spatial lobe could be obtained similarly.

After designing the analog precoding matrices, digital precoding matrices could be obtained by simply SVD for the effective channel since there is no constant amplitude constraint on digital precoding mareices. The effective channel for each spatial lobe could be calculated as
\begin{equation}\label{}
\begin{split}
{{\bf{H}}_{{{\rm{eq}}_i}}} &= {\bf{W}}_{{{\rm{RF}}_i}}^* {\bf{H}}{{\bf{F}}_{{{\rm{RF}}_i}}},\\
&={{\bf{U}}_{{{\rm{eq}}_i}}}{{\bf{\Sigma}}_{{{\rm{eq}}_i}}}{{\bf{V}}_{{{\rm{eq}}_i}}^*}, i=1,2,...,P.
\end{split}
\end{equation}
Therefore, the digital precoding matrices for the $i_{th}$ spatial lobe could be computed as
\begin{equation}\label{}
{{\bf{F}}_{{{\rm{BB}}_i}}} = {\bf{V}}_{{{\rm{eq}}_i}},
\end{equation}
and
\begin{equation}\label{}
{{\bf{W}}_{{{\rm{BB}}_i}}} = {{\bf{U}}_{{{\rm{eq}}_i}}}.
\end{equation}
The total digital precoding matrices for the transmitter and receiver are the block diagonal concatenation of ${{\bf{F}}_{{{\rm{BB}}_i}}}$ and ${{\bf{W}}_{{{\rm{BB}}_i}}}$, respectively. Finally, we normalize the total digital precoding matrix to satisfy the power constraint at the transmitter. The proposed NUQ-HYP-Full scheme for the full-connected structure is summarized in {\bf{Algorithm 2}}.


\section{Non-Uniform Quantization Codebook Based Hybrid Precoding for the Sub-connected Structure}
Different from the full-connected structure, each RF chain in the sub-connected structure is only connected to a subset of antennas or a subarray, which dramatically reduces the number of phase shifters, as shown in Fig.1(b) \cite{15Singh2015, 16Gao2016, 17Park2017}. Therefore, the sub-connected structure is more energy efficient than the full-connected structure. For simplicity but without loss of generality, we assume the transmitter and the receiver contain the same number of RF chains, i.e., ${N_{{\rm{RF}}}}={N_{{\rm{RF}}}^{\rm t}}={N_{{\rm{RF}}}^{\rm r}}$.

Similar as the hybrid precoding design in the full-connected structure, we also decouple the design of the analog and digital precoding matrices.
We focus on the design of the analog precoding matrices with pre-defined quantization codebooks in the sub-connected structure, since the digital precoding matrices could be directly obtained by simply SVD for the effective channel. To the best of the authors' knowledge, our work is the first attempt to utilize beamsteering based quantization codebooks to design hybrid precoding matrices in the sub-connected structure. 

\subsection{Non-Uniform Quantization Codebooks for the Sub-connected Structure}
\begin{algorithm}[t]
    \renewcommand{\algorithmicrequire}{\textbf{Input:}}
	\renewcommand{\algorithmicensure}{\textbf{Output:}}
	\caption{Non-Uniform Quantization (NUQ) Codebooks for the Sub-connected Structure}
    \label{alg:3}
	\begin{algorithmic}[1]
        \REQUIRE 
        $ N_{\rm{t}}$, $N_{\rm{r}}$, $N_{\rm{RF}}$,  $b$
        \ENSURE $\widehat{{\bf{A}}}_{\rm t}^{{\rm{quant}}}$, $\widehat{{\bf{A}}}_{\rm r}^{{\rm{quant}}}$
        \STATE $ N_{\rm{t}}^{sub}=N_{\rm{t}}/N_{\rm{RF}}$, $ N_{\rm{r}}^{sub}=N_{\rm{r}}/N_{\rm{RF}}$
        \FOR {$n\leq N_{\rm{RF}}$}
        \STATE $\widehat{\bf{A}}_{\rm{index}}^{\rm{t}}(n,:)=\Lambda_{{\rm{t}}{n}}=[(n-1)N_{\rm{t}}^{sub}, ..., n(N_{\rm{t}}^{sub})-1]$
        \STATE $\widehat{\bf{A}}_{\rm{index}}^{\rm{r}}(n,:)=\Lambda_{{\rm{r}}{n}}=[(n-1)N_{\rm{r}}^{sub}, ..., n(N_{\rm{r}}^{sub})-1]$
        \ENDFOR
        \FOR {$i\leq P$}
        \STATE ${{\cal{CV}}_{i}^{{\rm{quant}}}}=[{\widetilde{\theta}}_i-{\theta^{{\rm range}}_i}/2, {\widetilde{\theta}}_i+{\theta^{{\rm range}}_i}/2]$
        \STATE $\Delta_\theta=2\pi/(2^b)$
        \STATE ${\Theta}_{ i}^{{\rm{quant}}}=\Big[{\widetilde{\theta}}_i-\frac{\theta^{{\rm range}}_i}{2}, {\widetilde{\theta}}_i-\frac{\theta^{{\rm range}}_i}{2}+\Delta_\theta, ..., {\widetilde{\theta}}_i+\frac{\theta^{{\rm range}}_i}{2}\Big]$
        \FOR {$q\leq Q$}
        \FOR {$m\leq {\cal L}(\Theta_{i}^{{\rm{quant}}})$}
        \STATE ${\widehat{{\bf{A}}}_{{\rm t}{iq}}^{{\rm{quant}}}}(:,m)=\sqrt {\frac{1}{N_{\rm{t}}^{\rm sub}}}e^{(j\pi\Lambda_{{\rm t}((i-1)*Q+q)}*{\rm sin}(\Theta_{i}^{{\rm{quant}}}(m)))} $
        \STATE ${\widehat{{\bf{A}}}_{{\rm r}{iq}}^{{\rm{quant}}}}(:,m)=\sqrt {\frac{1}{N_{\rm{r}}^{\rm sub}}}e^{(j\pi\Lambda_{{\rm r}((i-1)*Q+q)}*{\rm sin}(\Theta_{i}^{{\rm{quant}}}(m)))} $
        \ENDFOR
        \ENDFOR
        \STATE ${\widehat{{\bf{A}}}_{{\rm t}{i}}^{{\rm{quant}}}}={\rm{blkdiag}}({\widehat{{\bf{A}}}_{{\rm t}{i1}}^{{\rm{quant}}}}, {\widehat{{\bf{A}}}_{{\rm t}{i2}}^{{\rm{quant}}}},...,{\widehat{{\bf{A}}}_{{\rm t}{iQ}}^{{\rm{quant}}}})$
        \STATE ${\widehat{{\bf{A}}}_{{\rm r}{i}}^{{\rm{quant}}}}={\rm{blkdiag}}({\widehat{{\bf{A}}}_{{\rm r}{i1}}^{{\rm{quant}}}}, {\widehat{{\bf{A}}}_{{\rm r}{i2}}^{{\rm{quant}}}},...,{\widehat{{\bf{A}}}_{{\rm r}{iQ}}^{{\rm{quant}}}})$
        \ENDFOR
        \STATE ${\widehat{{\bf{A}}}_{\rm t}^{{\rm{quant}}}}={\rm{blkdiag}}({\widehat{{\bf{A}}}_{\rm t1}^{{\rm{quant}}}}, {\widehat{{\bf{A}}}_{\rm t2}^{{\rm{quant}}}},...,{\widehat{{\bf{A}}}_{{\rm t}{P}}^{{\rm{quant}}}})$
        \STATE ${\widehat{{\bf{A}}}_{\rm r}^{{\rm{quant}}}}={\rm{blkdiag}}({\widehat{{\bf{A}}}_{\rm r1}^{{\rm{quant}}}}, {\widehat{{\bf{A}}}_{\rm r2}^{{\rm{quant}}}},...,{\widehat{{\bf{A}}}_{{\rm r}{P}}^{{\rm{quant}}}})$
    \end{algorithmic}
\end{algorithm}

In the design of the NUQ codebooks for the sub-connected structure, the number of the total RF chains equals the number of the subarray, i.e.,
${N_{{\rm{RF}}}}={N_{{\rm{sub}}}}$, which means the number of the antennas in each subarray $N_{\rm{t}}^{{\rm{sub}}}$ equals $N_{\rm{t}}/{N_{{\rm{RF}}}}$.
At the transmitter, let the total antenna indexes be $[{0,1,...,N_{\rm{t}}-1}]$ and $\Lambda_{{\rm{t}}k}$ denote the partitioned subset of antenna indexes connected to the $k_{th}$ subarray, such as
\begin{equation}\label{4.5}
\begin{split}
\Lambda_{{\rm{t}}1}&=[{0,...,N_{\rm{t}}^{{\rm{sub}}}-1}],\\
\Lambda_{{\rm{t}}2}&=[{N_{\rm{t}}^{{\rm{sub}}},...,2N_{\rm{t}}^{{\rm{sub}}}-1}],\\
\vdots \\
\Lambda_{{\rm{t}}{N_{{\rm{RF}}}}}&=[{({N_{{\rm{RF}}}}-1)N_{\rm{t}}^{{\rm{sub}}},...,{N_{{\rm{RF}}}}N_{\rm{t}}^{{\rm{sub}}}-1}].
\end{split}
\end{equation}
The total antenna index matrix at the transmitter for the sub-connected structure is
\begin{equation}\label{4.6}
\widehat{\bf{A}}_{\rm{index}}^{\rm{t}}=[\Lambda_{{\rm{t}}1}; \Lambda_{{\rm{t}}2};...; \Lambda_{{{\rm{t}}{N_{{\rm{RF}}}}}}].
\end{equation}
The total antenna index matrix at the receiver $\widehat{\bf{A}}_{\rm{index}}^{\rm{r}}$ for the sub-connected structure could be obtained similarly.

For the limited scattering millimeter wave channel, the number of total effective paths is usually very small  \cite{16Gao2016, 20Alkhateeb2014}. Thus, we make an assumption that the total number of effective paths is no more than the total number of RF chains, i.e., $PQ \leq N_{\rm RF}$.
In the meantime, since each subarray is only connected to one RF chain, each subarray is arranged to precoding for one path separately.
From (\ref{4.5}) and (\ref{4.6}), we could find the analog precoding matrices are in a block diagonal form, which could be depicted as
\begin{equation}\label{}
{ {\widehat{{\bf{F}}}_{{\rm{RF}}}}}= \left[ \begin{array}{l}
{ {\widehat{{\bf{F}}}_{{\rm{RF}}_1}}}\\
{\kern 19pt} { {\widehat{{\bf{F}}}_{{\rm{RF_2}}}}}\\
{\kern 38pt}  \ddots  \\
{\kern 52pt} {\widehat{{\bf{F}}}_{{\rm{RF}}_P}}\\
{\kern 68pt} {\widehat{{\bf{F}}}_{{\rm{RF}}_{(N_{\rm RF}-PQ)}}}\\
\end{array} \right],
\end{equation}
where
\begin{equation}\label{}
{ {\widehat{{\bf{F}}}_{{\rm{RF}}_i}}}= \left[ \begin{array}{l}
{ {{\bf{a}}_{{{i1}}}}}\\
{\kern 19pt} { {{\bf{a}}_{{{i2}}}}}\\
{\kern 38pt}  \ddots  \\
{\kern 52pt} {{\bf{a}}_{{{iQ}}}}
\end{array} \right], i=1,2,...,P,
\end{equation}
is the analog precoding matrix for the $i_{th}$ spatial lobe and the elements of ${\widehat{{\bf{F}}}_{{\rm{RF}}_{(N_{\rm RF}-PQ)}}}$ are 0.

Similar to (\ref{3.7})-(\ref{3.12}), the quantized angles coverage, quantized angle accuracy and quantization angles for the $i_{th}$ spatial lobe are ${{\cal{CV}}_{i}^{{\rm{quant}}}}$, $\Delta_\theta$ and $\Theta_{i}^{{\rm{quant}}}$, respectively. Therefore, the non-uniform quantization codebooks for the sub-connected structure can be designed as
\begin{equation}\label{}
\widehat{{{\bf{A}}}}_{\rm t}^{{\rm{quant}}}= \left[ \begin{array}{l}
\widehat{{{\bf{A}}}}_{\rm t1}^{{\rm{quant}}}\\
{\kern 19pt} \widehat{{{\bf{A}}}}_{\rm t2}^{{\rm{quant}}}\\
{\kern 38pt}  \ddots  \\
{\kern 52pt} \widehat{{{\bf{A}}}}_{{\rm{t}}P}^{{\rm{quant}}}
\end{array} \right],
\end{equation}
where
\begin{equation}\label{}
\widehat{{{\bf{A}}}}_{{\rm{t}}i}^{{\rm{quant}}}= \left[ \begin{array}{l}
\widehat{{{\bf{A}}}}_{{\rm{t}}i1}^{{\rm{quant}}}\\
{\kern 19pt} \widehat{{{\bf{A}}}}_{{\rm{t}}i2}^{{\rm{quant}}}\\
{\kern 38pt}  \ddots  \\
{\kern 52pt} \widehat{{{\bf{A}}}}_{{\rm{t}}iQ}^{{\rm{quant}}}
\end{array} \right], i=1,2,...,P,
\end{equation}
in which
\begin{equation}\label{}
{\widehat{{\bf{A}}}_{{\rm{t}}iq}^{{\rm{quant}}}}(:,m)=\sqrt {\frac{1}{N_{\rm{t}}^{\rm{sub}}}}e^{(j\pi \Lambda_{t((i-1)*Q+q)}*{\rm sin}(\Theta_{i}^{{\rm{quant}}}(m)))},
\end{equation}
where $q=1,2,...,Q$, $m=1,2,...,{\cal L}(\Theta_{i}^{{\rm{quant}}})$. The non-uniform quantization codebooks for the receiver $\widehat{{{\bf{A}}}}_{\rm r}^{{\rm{quant}}}$ could be designed similarly. In {\bf{Algorithm} \ref{alg:3}}, we summarize the process of constructing the non-uniform quantization codebooks for the sub-connected structure.

\subsection{NUQ-Based Hybrid Precoding for the Sub-connected Structure}
\begin{algorithm*}[t]
    \renewcommand{\algorithmicrequire}{\textbf{Input:}}
	\renewcommand{\algorithmicensure}{\textbf{Output:}}
	\caption{NUQ Codebook Based Hybrid Precoding for the Sub-connected Structure, i.e., NUQ-HYP-Sub}
    \label{alg:4}
	\begin{algorithmic}[1]
        \REQUIRE 
        $ {{\bf{A}}_{\rm{t}}}$, ${{\bf{A}}_{\rm{r}}}$, ${\widehat{{\bf{A}}}_{\rm t}^{{\rm{quant}}}}$, ${\widehat{{\bf{A}}}_{\rm r}^{{\rm{quant}}}}$
		\ENSURE ${\widehat{{\bf{F}}}_{{\rm{RF}}}}$, ${\widehat{{\bf{F}}}_{{\rm{BB}}}},{\widehat{{\bf{W}}}_{{\rm{RF}}}}$, ${\widehat{{\bf{W}}}_{{\rm{BB}}}}$
        \STATE {\bf {First stage}}: Analog precoding matrices design for both transmitter and receiver
		\FOR {$i\leq P$}
		\FOR {$q\leq Q$}
        \STATE ${\bf{F}}_{{\rm{res}}}={{\bf{A}}_{{\rm t}{i}}}(:,q)$
        \STATE ${\bf{\Psi}}  = \widehat{{\bf{A}}}_{\rm{t}}^{{\rm{quant*}}}(:,\frac{\cal L({\widehat{{\bf{A}}}_{\rm t}^{{\rm{quant}}}})}{(PQ)}((i-1)*Q+q-1)+1:\frac{\cal L({\widehat{{\bf{A}}}_{\rm t}^{{\rm{quant}}}})}{(PQ)}*((i-1)*Q+q)){{\bf{F}}_{{\rm{res}}}}$
        \STATE $k = {\rm{argma}}{{\rm{x}}_{l = 1,...,{\rm{Q}}}}{\rm{diag}}({\bf{\Psi }}{{\bf{\Psi }}^*})$
        \STATE ${\widehat{{\bf{F}}}_{{\rm{RF}}_{iq}}} = {\widehat{{\bf{A}}}_{\rm{t}}^{{\rm{quant}}}}(:, \frac{(\cal L({\widehat{{\bf{A}}}_{\rm t}^{{\rm{quant}}}})}{(PQ)}*((i-1)*Q+q-1)+k)$
		\ENDFOR
        \STATE${\widehat{{\bf{F}}}_{{\rm{RF}}_i}} = [{\widehat{{\bf{F}}}_{{\rm{RF}}_{i1}}},{\widehat{{\bf{F}}}_{{\rm{RF}}_{i2}}},...,{\widehat{{\bf{F}}}_{{\rm{RF}}_{iQ}}}]{\kern 1pt} {\kern 1pt}$
        \ENDFOR
        \STATE${\widehat{{\bf{F}}}_{{\rm{RF}}}} = [{\widehat{{\bf{F}}}_{{\rm{RF_1}}}},{\widehat{{\bf{F}}}_{{\rm{RF_2}}}},...,{\widehat{{\bf{F}}}_{{\rm{RF}}_P}}]{\kern 1pt} {\kern 1pt}$
        \STATE We could obtain ${\widehat{{\bf{W}}}_{{\rm{RF}}}}$ in the similar way\\
        \STATE ${\widehat{{\bf{W}}}_{{\rm{RF}}}} = [{\widehat{{\bf{W}}}_{{\rm{RF_1}}}},{\widehat{{\bf{W}}}_{{\rm{RF_2}}}},...,{\widehat{{\bf{W}}}_{{\rm{RF}}_P}}]$
        \STATE Compute the effective channels, ${\widehat{{\bf{H}}}_{{\rm{eq}}_i}} = \widehat{{\bf{W}}}_{{\rm{RF}}_i}^{\rm{*}}{\bf{H}}{\widehat{{\bf{F}}}_{{\rm{RF}}_i}},i=1,2,...,P$
        \STATE {\bf {Second stage}}: Digital precoding matrices design for both transmitter and receiver
        \STATE Compute the SVD of each effective channel ${\widehat{{\bf{H}}}_{{\rm{eq}}_i}}$,
        ${\bf{W}}_{{\rm{RF}}_i}^{\rm{*}}{\bf{H}}{{\bf{F}}_{{\rm{RF}}_i}} = {{\bf{U}}_{{{\rm{eq}}}_i}}{{\bf{\Sigma}}_{{{\rm{eq}}}_i}}{{\bf{V}}_{{{\rm{eq}}_i}}^*}$
        \STATE ${{\bf{F}}_{{\rm{BB}}_i}} = {\bf{V}}_{{{\rm{eq}}_i}}$,${{\bf{W}}_{{\rm{BB}}_i}} = {{\bf{U}}_{{{\rm{eq}}_i}}}$
        \STATE ${\widehat{{\bf{F}}}_{{\rm{BB}}}} = {\rm{blkdiag}}({\widehat{{\bf{F}}}_{{\rm{BB_1}}}},{\widehat{{\bf{F}}}_{{\rm{BB_2}}}},...,{\widehat{{\bf{F}}}_{{\rm{BB}}_P}})$,
        ${\widehat{{\bf{W}}}_{{\rm{BB}}}} = {\rm{blkdiag}}({\widehat{{\bf{W}}}_{{\rm{BB_1}}}},{\widehat{{\bf{W}}}_{{\rm{BB_2}}}},...,{\widehat{{\bf{W}}}_{{\rm{BB}}_P}})$
        \STATE Normalize the digital precoding matrix at the transmitter, ${\widehat{{\bf{F}}}_{{\rm{BB}}}} = \sqrt {{N_s}} \frac{{{\widehat{{\bf{F}}}_{{\rm{BB}}}}}}{{{{\left\| {{\widehat{{\bf{F}}}_{{\rm{RF}}}}{\widehat{{\bf{F}}}_{{\rm{BB}}}}} \right\|}_F}}}$
    \end{algorithmic}
\end{algorithm*}
The designing of hybrid precoding matrices for the energy-efficient sub-connected structure is an attractive topic for recent works in millimeter wave MIMO systems, e.g., the SIC-based hybrid precoding scheme which decomposes the total optimization problem into several simple sub optimization problems and achieves the near-optimal performance.
However, most of the prior works such as SIC-based hybrid precoding scheme did not consider the limited feedback problem which is also an important issue for the sub-connected structure.
Moreover, to the best of our knowledge, there has been no beamsteering codebooks based hybrid precoding scheme for the sub-connected architecture in limited feedback millimeter wave MIMO systems.
In this subsection, based on the proposed non-quantization codebooks described in Section IV.A, we propose a NUQ codebook based hybrid precoding scheme for the sub-connected structure, which is summarized in {\bf{Algorithm} \ref{alg:4}}.

In the analog precoding matrices designing stage, we quantize the array response vector corresponding to each path one by one to maintain the block diagonal form of the analog precoding matrix. For the $q_{th}$ subpath in the $i_{th}$ spatial lobe, the reference precoding vector is ${{\bf{A}}_{{\rm{t}}i}}(:,q)$ and the quantized candidate matrix is defined as
\begin{equation}\label{}
\small
\widehat{{{\bf{A}}}}_{{\rm{t}}iq}^{{\rm{can}}}= \left[ \begin{array}{l}
{\bf 0}_{((i-1)Q+q-1)N_{\rm{t}}^{\rm{sub}}\times {\cal L}(\Theta_{i}^{{\rm{quant}}})}\\
{\kern 45pt}\widehat{{{\bf{A}}}}_{{\rm{t}}iq}^{{\rm{quant}}}\\
{\bf 0}_{(N_{\rm{t}}-((i-1)Q+q)N_{\rm{t}}^{\rm{sub}})\times {\cal L}(\Theta_{i}^{{\rm{quant}}})}
\end{array} \right],
\end{equation}
which is actually the $\frac{\cal L({\widehat{{\bf{A}}}_{\rm t}^{{\rm{quant}}}})}{(PQ)}((i-1)*Q+q-1)+1$ to $\frac{\cal L({\widehat{{\bf{A}}}_{\rm t}^{{\rm{quant}}}})}{(PQ)}*((i-1)*Q+q)$ columns of $\widehat{{{\bf{A}}}}_{\rm t}^{{\rm{quant}}}$, where $\frac{\cal L({\widehat{{\bf{A}}}_{\rm t}^{{\rm{quant}}}})}{(PQ)}$ is equal to ${\cal L}(\Theta_{i}^{{\rm{quant}}})$.
Therefore, the analog precoding design problem for the $q_{th}$ subpath in the $i_{th}$ spatial lobe can be formulated as
\begin{equation}\label{4.b2}
\begin{array}{l}
\widehat{{\bf{F}}}_{{\rm{RF}}_{iq}}^{{\rm{opt}}} = \mathop {{\rm{arg}}{\kern 1pt} {\kern 1pt} {\rm{min}}} {\left\| {\widehat{{\bf{F}}}_{{\rm{RF}}_{iq}}} - {{\bf{A}}_{{\rm{t}}i}}(:,q) \right\|_F},{\kern 1pt} {\kern 1pt} {\kern 1pt} {\kern 1pt} {\kern 1pt} {\kern 1pt} {\kern 1pt} {\kern 1pt} {\kern 1pt} {\kern 1pt} {\kern 1pt} {\kern 1pt} {\kern 1pt} {\kern 1pt} {\kern 1pt} {\kern 1pt} \\
 {\kern 58pt} {\rm{s}}{\rm{.t}}{\rm{.}}{\kern 7pt}  {\widehat{{\bf{F}}}_{{\rm{RF}}_{iq}}} \in {\widehat{{\bf{A}}}_{{\rm t}iq}^{{\rm{can}}}}.
\end{array}
\end{equation}
 Similar to (\ref{3.c1}), (\ref{4.b2}) could be also directly solved by searching each entry of ${\widehat{{\bf{A}}}_{{\rm t}iq}^{{\rm{can}}}}$ to find the column vector which has the maximum projection along ${{\bf{A}}_{{\rm{t}}i}}(:,q)$. Thus, the total analog precoding matrices could be obtained by
\begin{equation}\label{}
{\widehat{{\bf{F}}}_{{\rm{RF}}}} = [{\widehat{{\bf{F}}}_{{\rm{RF_1}}}},{\widehat{{\bf{F}}}_{{\rm{RF_2}}}},...,{\widehat{{\bf{F}}}_{{\rm{RF}}_P}}],
\end{equation}
where
\begin{equation}\label{}
{\widehat{{\bf{F}}}_{{\rm{RF}}_i}} = [{\widehat{{\bf{F}}}_{{\rm{RF}}_{i1}}},{\widehat{{\bf{F}}}_{{\rm{RF}}_{i2}}},...,{\widehat{{\bf{F}}}_{{\rm{RF}}_{iQ}}}], i=1,2,...,P.
\end{equation}
The analog combining matrix ${\widehat{{\bf{W}}}_{{\rm{RF}}}}$ for the receiver could be obtained similarly.

Then, for the digital precoding designing stage, the steps are similar as has been described for the full-connected structure. Therefore, we omit the digital precoding designing steps for the sub-connected structure, which are presented in detail in steps 15-20 of {\bf{Algorithm} \ref{alg:4}}.

\section{Simulation Results}
In this section, the performances of the proposed NUQ-HYP-Full and NUQ-HYO-Sub schemes are evaluated. The uniform quantization based OMP scheme (marked as UQ-OMP) and SIC-based hybrid precoding scheme are adopted as the benchmarks for the full-connected and sub-connected structures, respectively. We consider a single-user MIMO system, where ULAs with $\lambda/2$ antenna spacing are equipped at both the transmitter and the receiver. According to the measurements in downtown Manhattan environment which is a typical urban environment, the frequency of the millimeter wave is set to be 28 GHz and the bandwidth is set to be 100 MHz \cite{21Samimi2014, 22Samimi2016, 23Rappaport2015}. The spatial lobes millimeter wave channel (as shown in (8)) which has sparsity property in the angular domain is adopted. For $P$ spatial lobes, the mean angles of spatial lobes are set to be ${\widetilde{\theta}}_i=\theta_{\rm CO}+\frac{2\pi}{P}(i-1), i=1,2,...,P$, where $\theta_{\rm CO}$ is a constant that is randomly selected within $[0, 2\pi]$ and the angle spreads are set as $\omega=\omega_1=...=\omega_P=\frac{\pi}{P}$. The angles of $Q$ subpaths in one spatial lobe are assumed to be randomly distributed and the corresoponding power obeys the Rayleigh distribution.
In order to maximize the reduction in the feedback overhead, the quantized angle range of each spatial lobe is set as ${\theta^{{\rm range}}_i}=\omega, i=1,2,...,P$, which means the total quantized angle range is narrowed by half.

\subsection{Full-connected hybrid precoding structure}
\begin{figure}[t]
\centering
\includegraphics[scale=0.65]{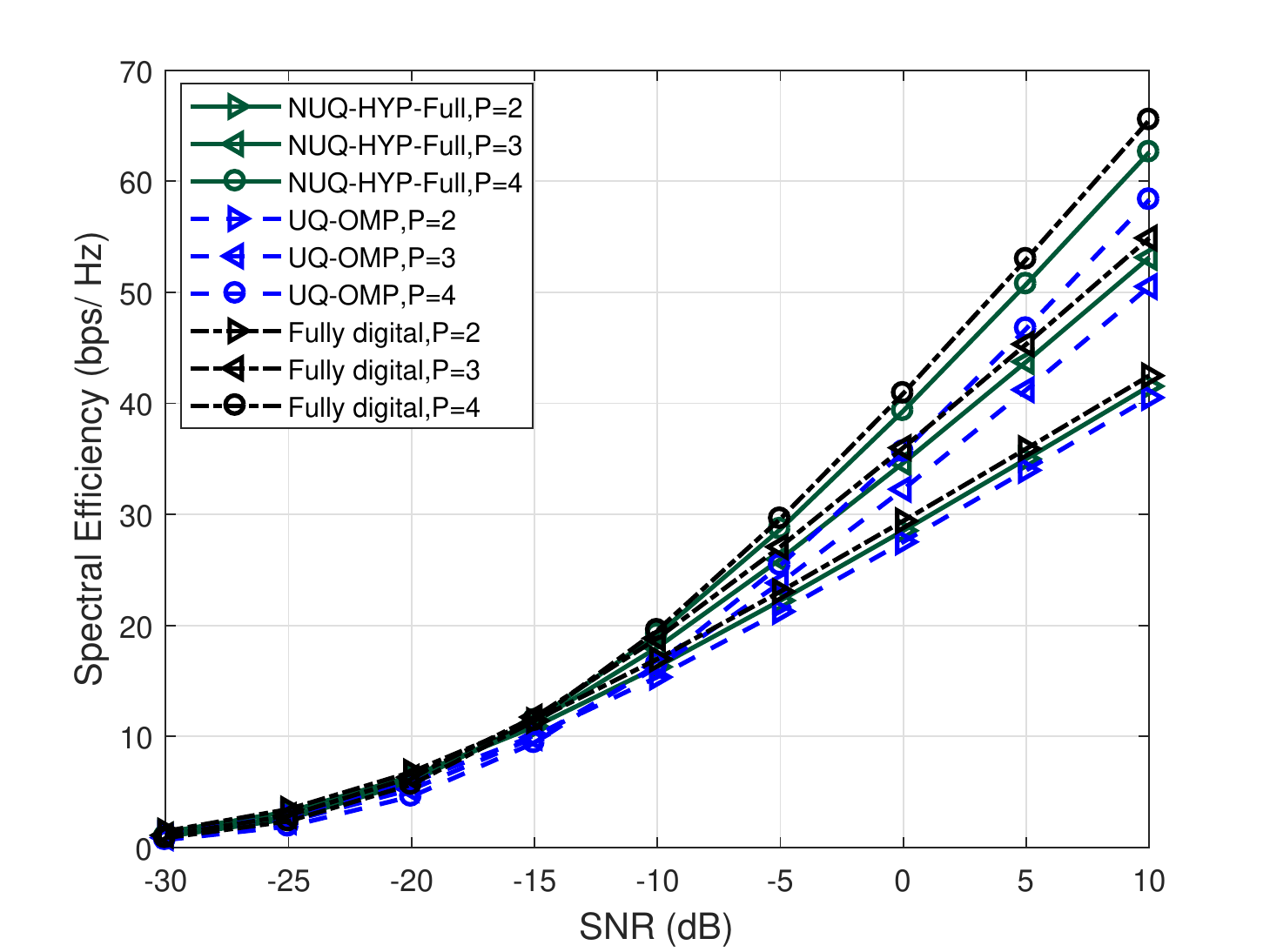}
\caption{Spectral efficiencies of NUQ-HYP-Full, OMP and fully digital precoding schemes with different numbers of spatial lobes, where $ N_{\rm{s}}=PQ$, $N_{\rm{t}}=144,N_{\rm{r}}=36, N_{\rm{RF}}^{\rm{t}}=N_{\rm{RF}}^{\rm{r}}=8, Q=2, b=8$.}
\label{figSprctralEfficiency_234-2}
\end{figure}

Fig. \ref{figSprctralEfficiency_234-2} and Fig. \ref{figSprctralEfficiency_2-234} compare the spectral efficiencies of the proposed NUQ-HYP-Full scheme, OMP precoding scheme and fully digital precoding scheme with different numbers of spatial lobes and subpaths, respectively, where $N_{\rm{t}}\times N_{\rm{r}} =144\times 36$, $N_{\rm{RF}}^{\rm{t}}=N_{\rm{RF}}^{\rm{r}}=8$ and $b=8$. Note that, the OMP precoding scheme we compare here is based on the uniform quantization beamsteering codebooks, which is slightly different from the Algorithm 1 in \cite{5Ayach2013}.
Since the total angle range that needs to be quantized is narrowed by half, the quantization accuracy of the NUQ-HYP-Full scheme could be increased by 1 bit. 
We observe that, for different numbers of spatial lobes and subpaths, the proposed NUQ-HYP-Full scheme always outperforms the OMP precoding scheme and achieves more than $95\%$ of the spectral efficiency achieved by the fully digital precoding scheme.

\begin{figure}[t]
\centering
\includegraphics[scale=0.65]{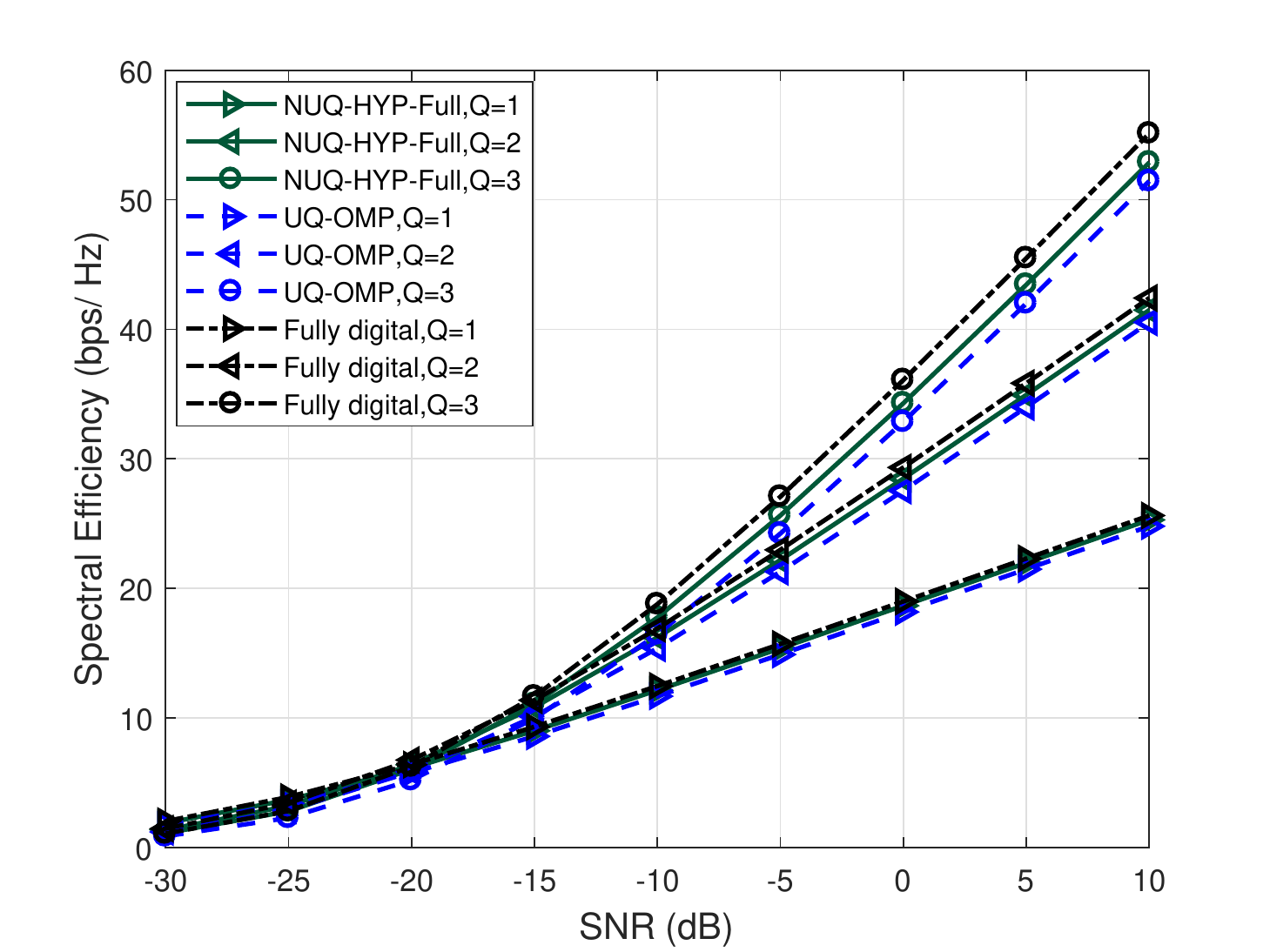}
\caption{Spectral efficiencies of NUQ-HYP-Full, OMP and fully digital precoding schemes with different numbers of subpaths, where $ N_{\rm{s}}=PQ$, $N_{\rm{t}}=144,N_{\rm{r}}=36, N_{\rm{RF}}^{\rm{t}}=N_{\rm{RF}}^{\rm{r}}=8, P=2, b=8$.}
\label{figSprctralEfficiency_2-234}
\end{figure}

\begin{figure}[t]
\centering
\includegraphics[scale=0.65]{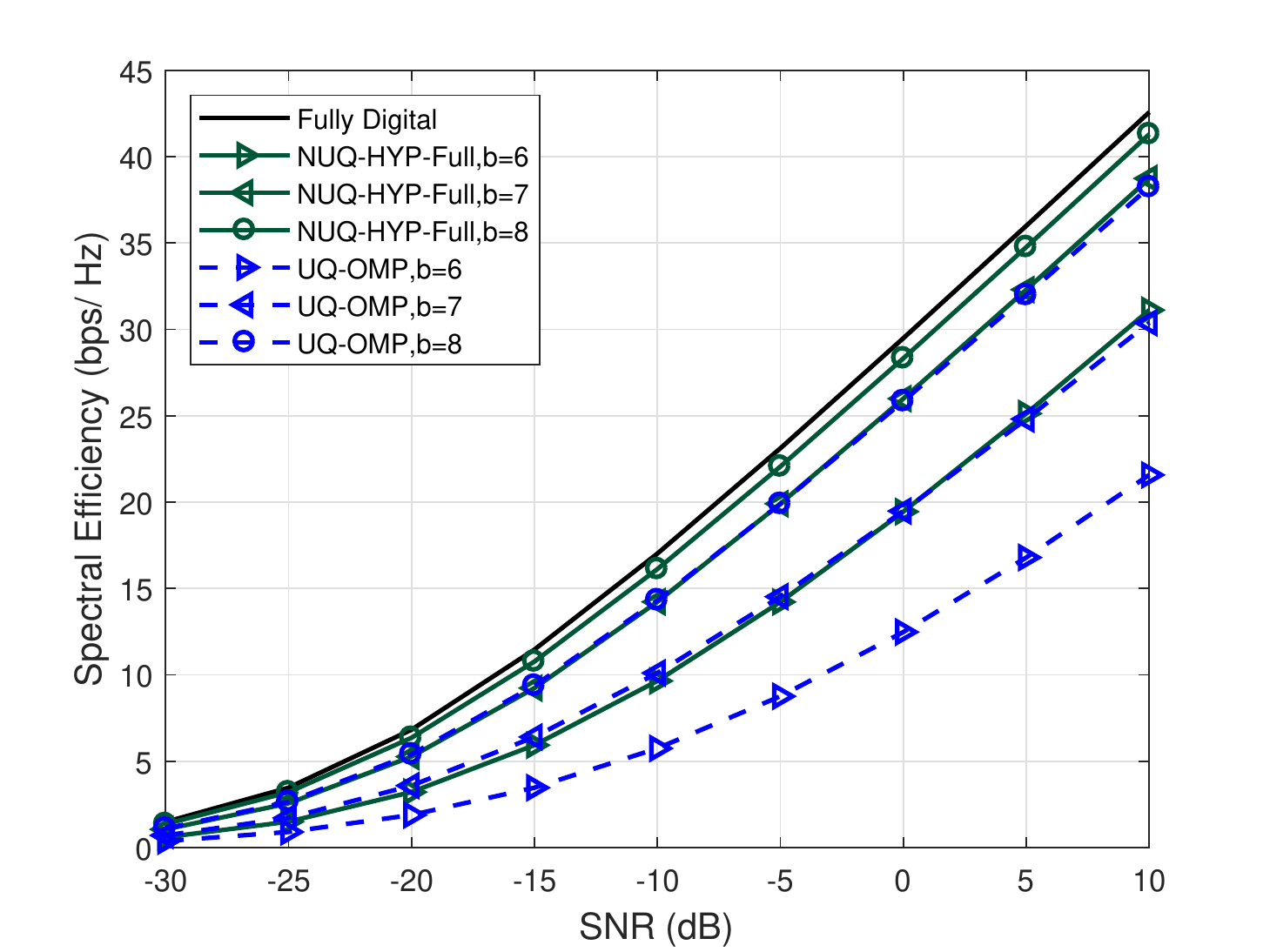}
\caption{Spectral efficiencies of NUQ-HYP-Fully, OMP and fully digital precoding schemes with different quantization bits, where $N_{\rm s}=PQ$, $N_t=144, N_{\rm{r}}=36, N_{\rm{RF}}^{\rm{t}}=N_{\rm{RF}}^{\rm{r}}=4, P=2$, Q=2.}
\label{figSprctralEfficiency_Qbits}
\end{figure}
In Fig. \ref{figSprctralEfficiency_Qbits}, the impact of the quantization bit on the spectral efficiency is presented, where $N_{\rm{t}}\times N_{\rm{r}} =144\times 36$, $N_{\rm{RF}}^{\rm{t}}=N_{\rm{RF}}^{\rm{r}}=4$, $P=2$ and $Q=2$. We observe that the spectral efficiencies of the NUQ-HYP-Full scheme are always higher than the UQ-OMP scheme with different numbers of quantization bits. Moreover, the smaller the number of quantization bits is, the larger the performance gap becomes. Specially, we also observe that the NUQ-HYP-Full scheme achieves similar spectral efficiencies as the UQ-OMP scheme, when the number of quantization bits is reduced by 1 (at least $12.5\%$ feedback overhead reduction for the given number of antennas). This phenomenon is in consistent with Proposition 1 presented in Section III.B. Furthermore,we observe that the number of quantization bits $b$ should satisfy $2^b\geq max({{N_{\rm t}},{N_{\rm r}}})$ to obtain good spectral efficiency for the UQ-OMP scheme that is based on the vector-by-vector UQ-based codebooks.

\begin{figure}[t]
\centering
\includegraphics[scale=0.65]{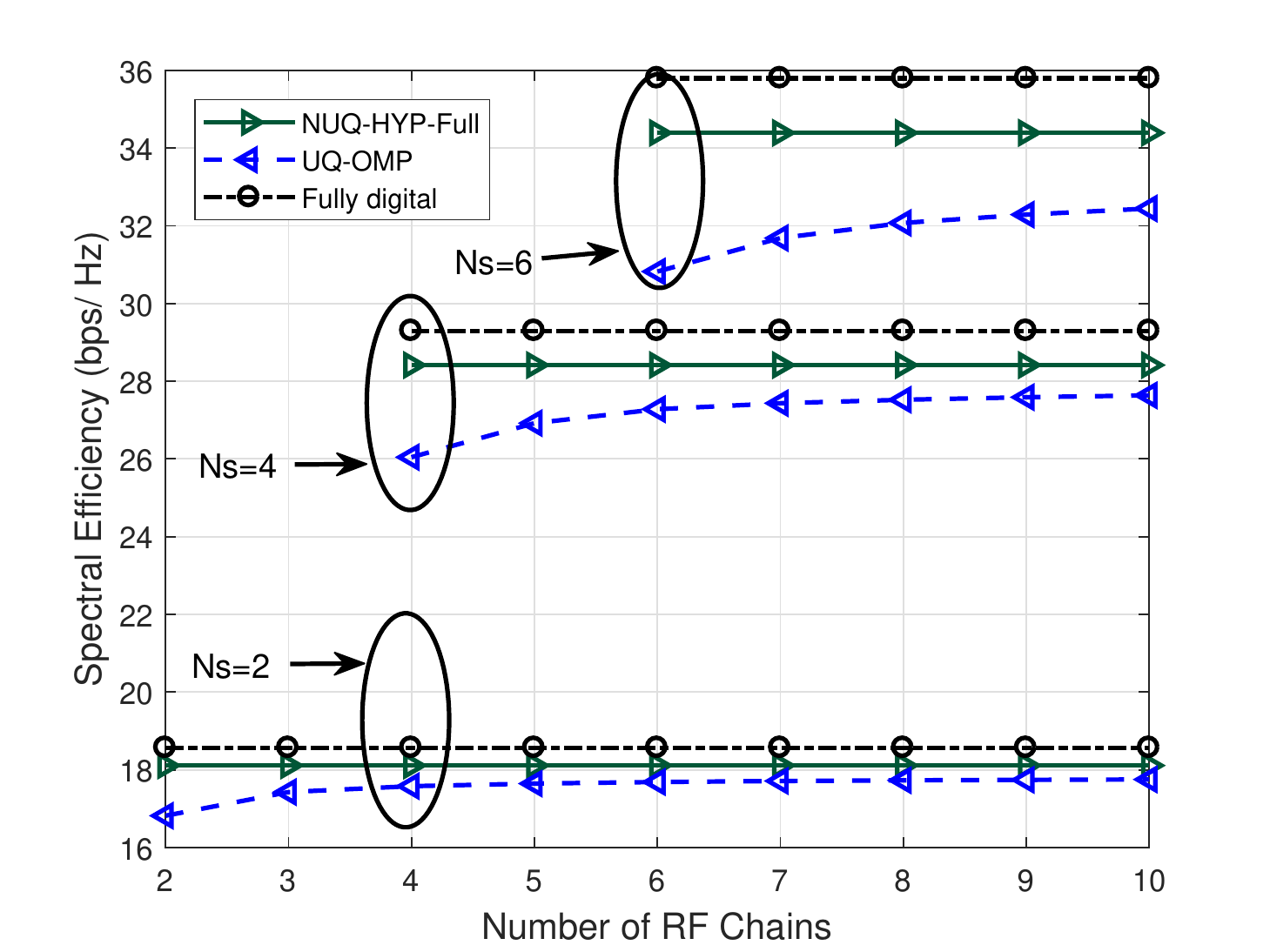}
\caption{Spectral efficiencies of NUQ-HYP-Full, OMP and fully digital precoding schemes with different numbers of RF chains, where $N_{\rm{RF}}^{\rm{t}}=N_{\rm{RF}}^{\rm{r}}, N_{\rm t}=144, N_{\rm{r}}=36, N_{\rm{s}}=PQ, Q=2$, SNR=0 dB, $b=8$.}
\label{figSprctralEfficiency_RF}
\end{figure}
Fig. \ref{figSprctralEfficiency_RF} shows the spectral efficiencies of different schemes with different numbers of RF chains, where $N_{\rm{t}}\times N_{\rm{r}} =144\times 36$ and the SNR is set to be 0 dB. We evaluate the cases when the number of subpaths $Q=2$ and the number of spatial lobes varies from 1 to 3. Since we only utilize $PQ$ RF chains to transmit and receiver signals, we observe that the spectral efficiencies of the proposed NUQ-HYP-Full scheme remain unchanged but are always higher than the spectral efficiencies achieved by the uniform quantization based OMP scheme.

\begin{figure}[t]
\centering
\includegraphics[scale=0.65]{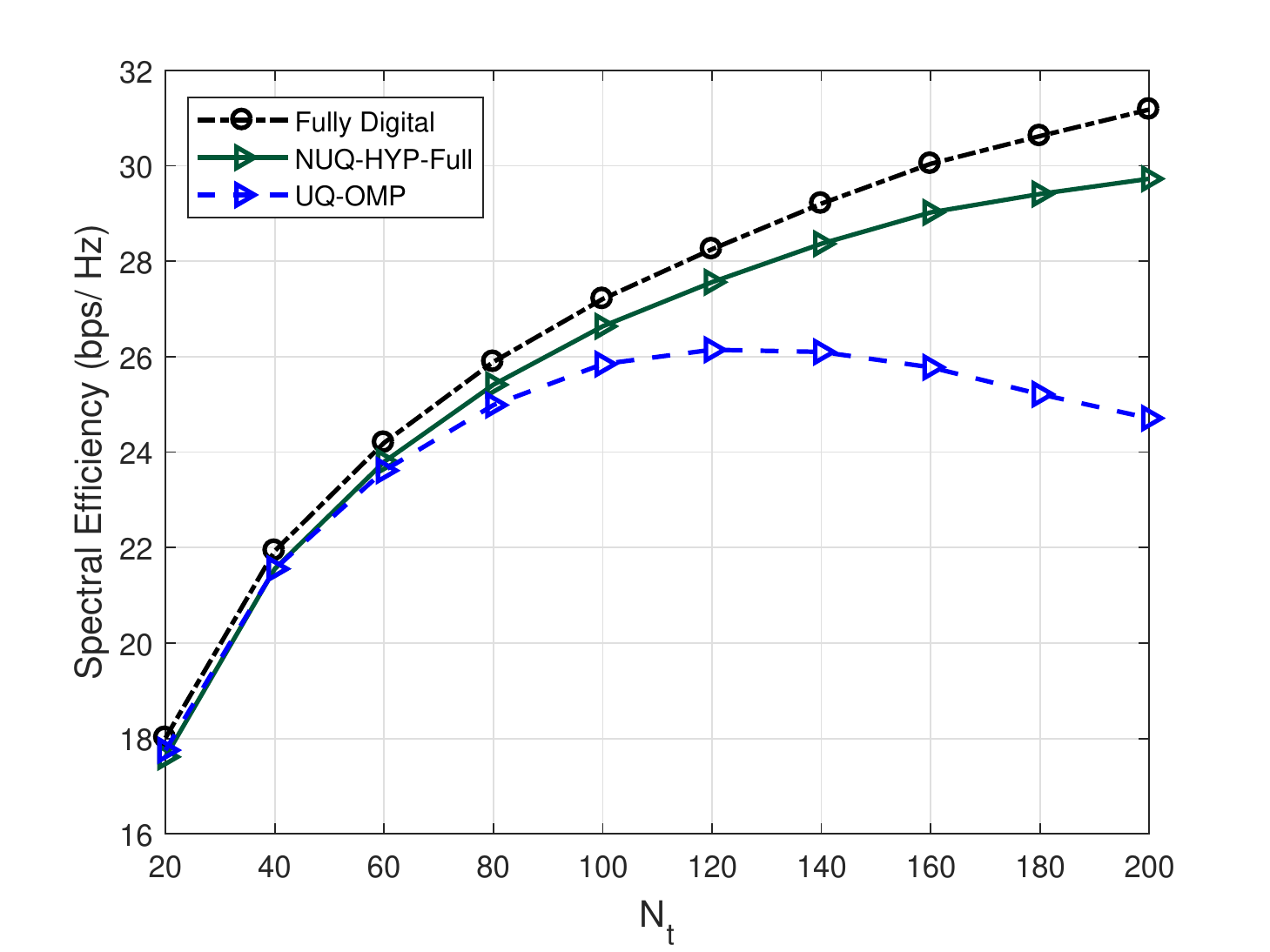}
\caption{Spectral efficiencies of NUQ-HYP-Full, OMP and fully digital precoding schemes with different numbers of transmitter's antennas $N_{\rm t}$, where $N_{\rm{r}}=36, N_{\rm{RF}}^{\rm{t}}=N_{\rm{RF}}^{\rm{r}}=N_{\rm{s}}=PQ, P=Q=2,$ SNR=0 dB, $b=8$.}
\label{figSprctralEfficiency_Nbs}
\end{figure}

\begin{figure}[t]
\centering
\includegraphics[scale=0.65]{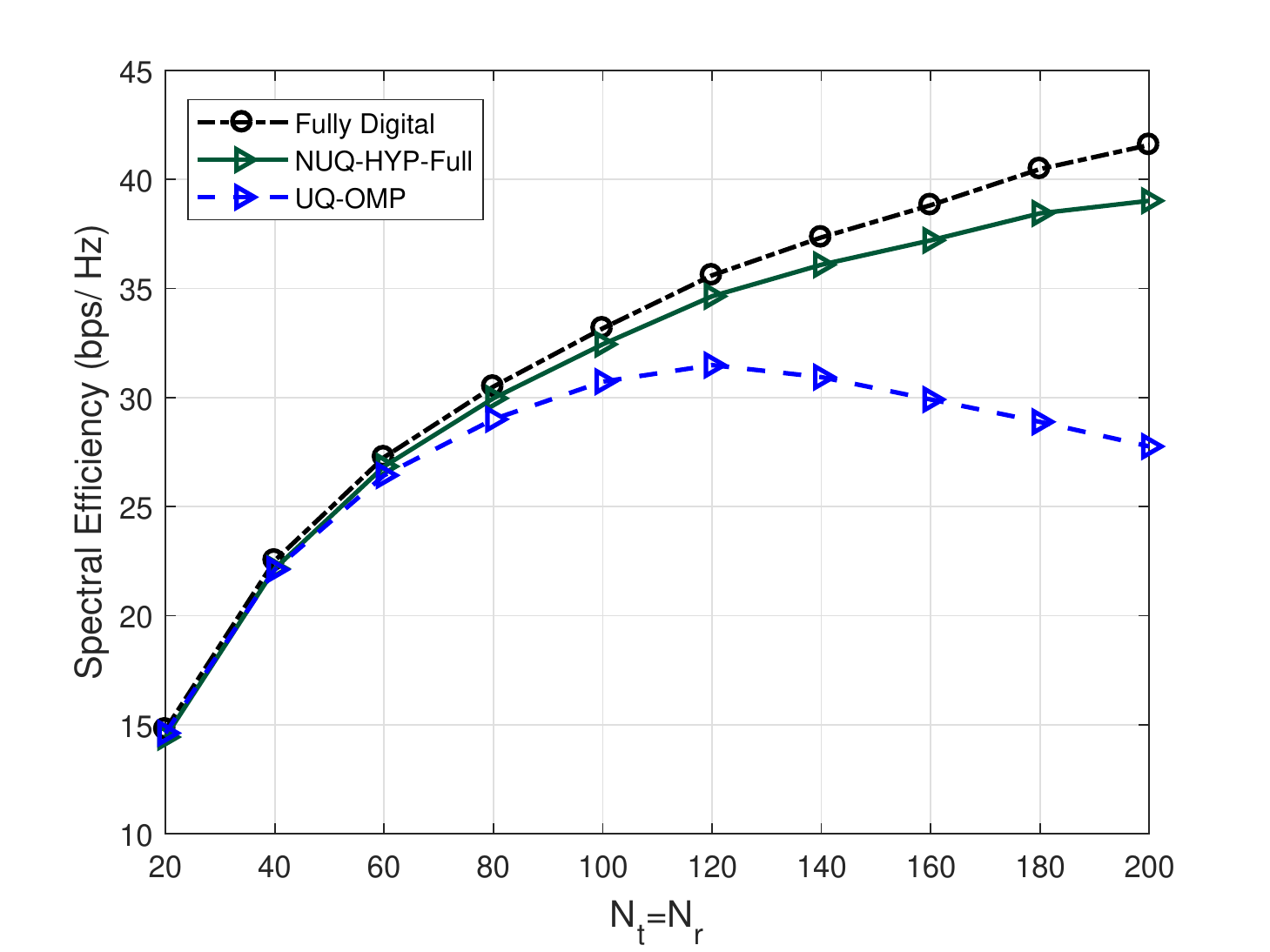}
\caption{Spectral efficiencies of NUQ-HYP-Full, OMP and fully digital precoding schemes with different numbers of $N_{\rm t}$ and $N_{\rm{r}}$ where $N_{\rm t}=N_{\rm{r}}$, $N_{\rm{RF}}^{\rm{t}}=N_{\rm{RF}}^{\rm{r}}=N_{\rm{s}}=PQ, P=Q=2,$ SNR=0 dB, $b=8$.}
\label{figSprctralEfficiency_NbsNms}
\end{figure}
Fig. \ref{figSprctralEfficiency_Nbs} and Fig. \ref{figSprctralEfficiency_NbsNms} show the spectral efficiencies of different schemes with different numbers antennas at the transmitter and both the transmitter and receiver, respectively, where $N_{\rm{RF}}^{\rm{t}}=N_{\rm{RF}}^{\rm{r}}=N_{\rm{s}}=PQ, P=Q=2,$ SNR=0 dB and the number of the quantization bit is set to be $b=8$. We observe that the proposed NUQ-HYP-Full scheme always outperforms the UQ-OMP scheme and achieves more than $95\%$ of the spectral efficiency achieved by the fully digital precoding scheme at even very large number of antennas. Moreover, we also observe that when the number of antennas becomes very large, the spectral efficiency of UQ-OMP scheme decreases, since the quantization bit $b$ is relatively not large enough.

\subsection{Sub-connected hybrid precoding structure}
For the sub-connected structure, the number of antennas are set as $N_{\rm{t}}\times N_{\rm{r}} =144\times 36$. In addition, the number of RF chains $N_{\rm{RF}}$ is set equal to $N_{\rm s}$, which is actually the worst case since $N_{\rm{RF}}$ should satisfy $N_{\rm{RF}}\geq N_{\rm s}$ to ensure the system could simultaneously transmit $N_{\rm s}$ data streams.

\begin{figure}[t]
\centering
\includegraphics[scale=0.65]{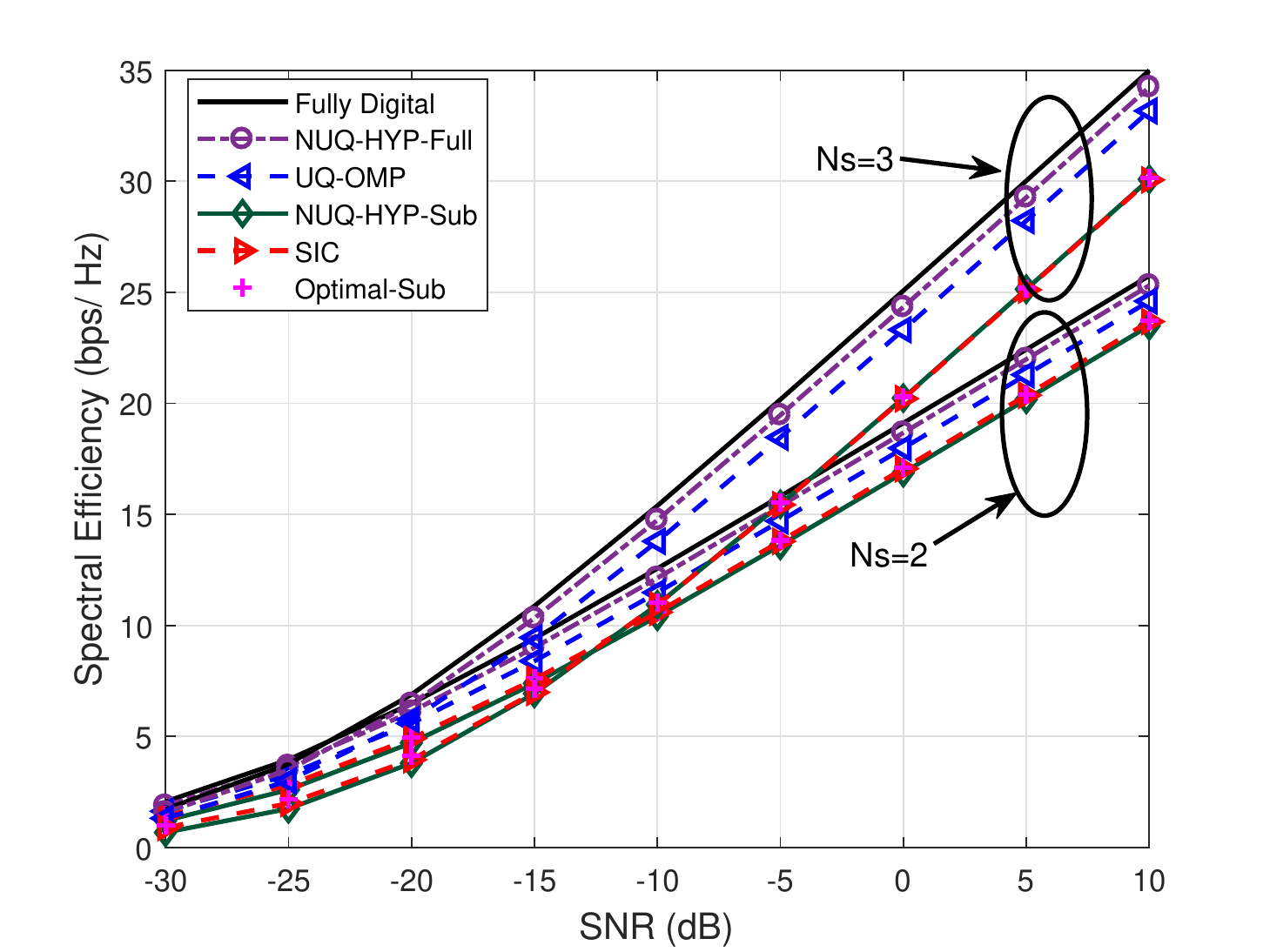}
\caption{Spectral efficiencies comparison with different numbers of spatial lobes for the sub-connected structure, where $N_{\rm t}=144, N_{\rm{r}}=36, N_{\rm{RF}}=N_{\rm{s}}=PQ, Q=1, b=6$.}
\label{figSprctralEfficiency_Sub_23-1}
\end{figure}

\begin{figure}[t]
\centering
\includegraphics[scale=0.65]{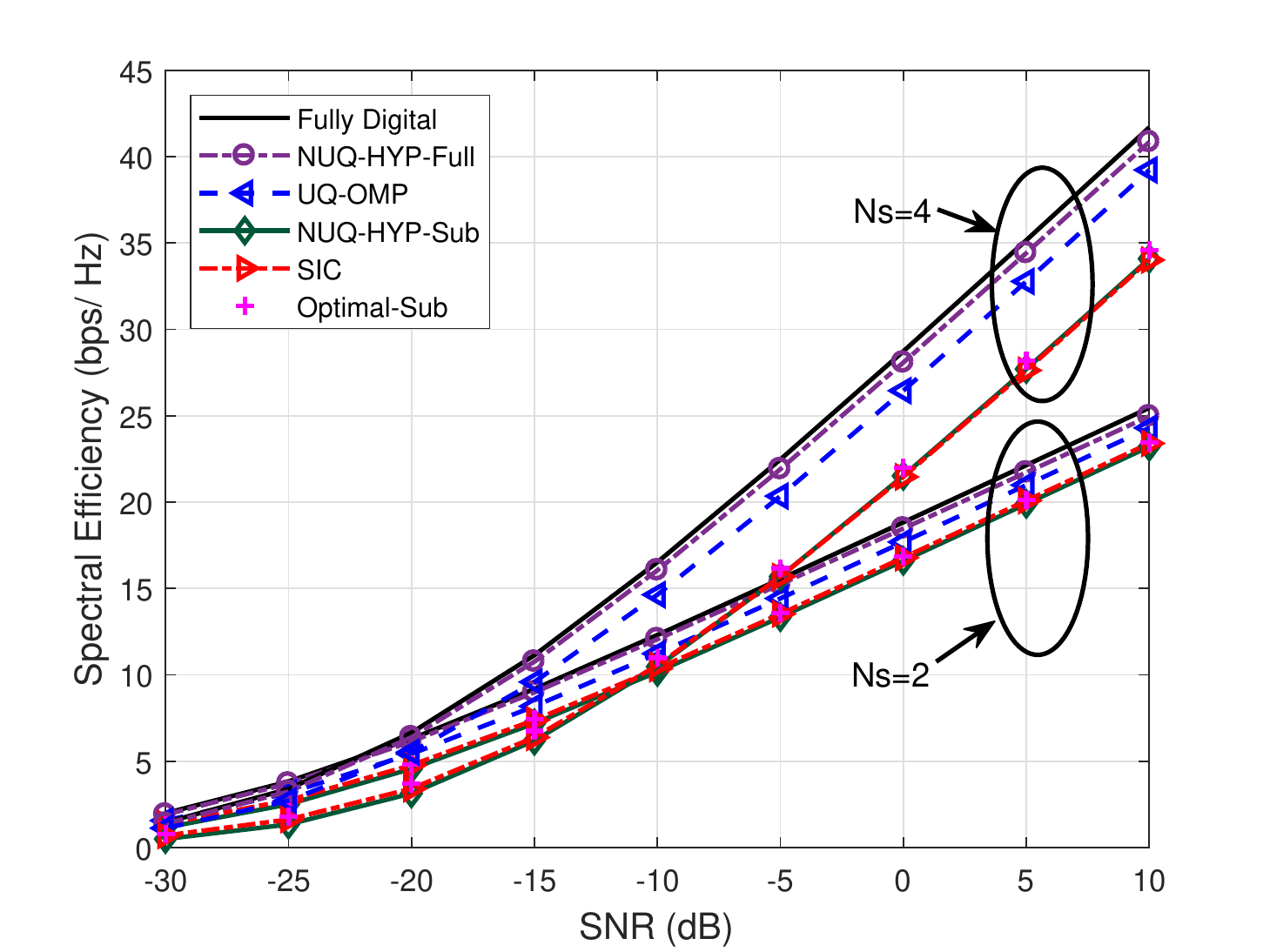}
\caption{Spectral efficiencies comparison with different numbers of subpaths for the sub-connected structure, where $N_{\rm t}=144, N_{\rm{r}}=36, N_{\rm{RF}}=N_{\rm{s}}=PQ, P=2, b=6$.}
\label{figSprctralEfficiency_Sub_2-12}
\end{figure}
Fig. \ref{figSprctralEfficiency_Sub_23-1} and Fig. \ref{figSprctralEfficiency_Sub_2-12} compare the spectral efficiencies of NUQ-HYP-Sub scheme, SIC-based hybrid precoding
scheme, optimal unconstrained precoding scheme for the sub-connected structure (marked as Optimal-Sub), UQ-OMP scheme and fully digital precoding scheme with different numbers of spatial lobes and subpaths, respectively. The optimal unconstrained precoding scheme for the sub-connected structure was detailedly described in \cite{16Gao2016}. We observe that the proposed NUQ-HYP-Sub scheme achieves similar spectral efficiencies as the SIC and Optimal-Sub schemes, and the performance gaps are less than $1\%$. In addition, Fig. \ref{figSprctralEfficiency_Sub_23-1} and Fig. \ref{figSprctralEfficiency_Sub_2-12} also show that, the proposed NUQ-HYP-Sub scheme achieves more than $87\%$ of the spectral efficiencies achieved by the UQ-OMP scheme for all cases.

\begin{figure}[t]
\centering
\includegraphics[scale=0.65]{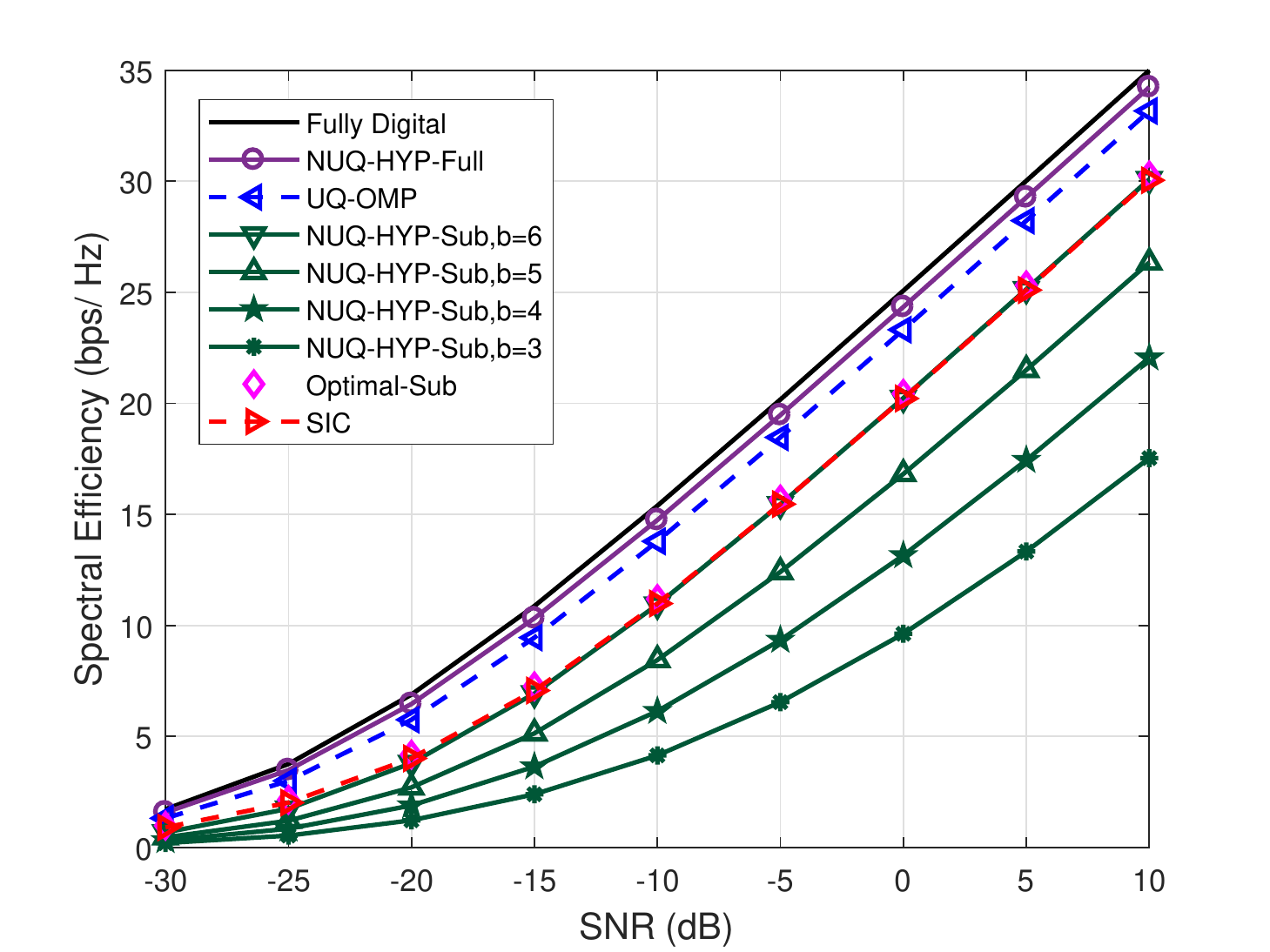}
\caption{Spectral efficiencies comparison with different quantization bits for the sub-connected structure, where $N_{\rm t}=144, N_{\rm{r}}=36, N_{\rm{RF}}=N_{\rm{s}}=PQ, P=3, Q=1$.}
\label{figSprctralEfficiency_Sub_Qbit}
\end{figure}
In Fig. \ref{figSprctralEfficiency_Sub_Qbit}, the impact of the quantization bit on the spectral efficiency for the sub-connected structure is shown, where the number of the spatial lobes and subpaths are set as $P=3$ and $Q=1$, respectively. We observe that as the number of quantization bits increases, the spectral efficiency gaps between the NUQ-HYP-Sub and Optimal-Sub schemes become smaller. Furthermore, we also observe that, to obtain the near-optimal spectral efficiency, the required number of quantization bits for the proposed NUQ-HYP-Sub scheme is at least 6. However, the UQ-OMP scheme needs at least 8 quantization bits to maintain the near-optimal performance at the full-connected structure. This is because only $N_{\rm t}^{{\rm sub}}=N_{\rm t}/{N_{{\rm RF}}}$ and $N_{\rm r}^{{\rm sub}}=N_{\rm r}/{N_{{\rm RF}}}$ antennas are utilized to transmit and receive the directional beams at the sub-connected structure, respectively. Therefore, we only need to make $2^b$ approach $max({N_{\rm t}^{{\rm sub}}, N_{\rm r}^{{\rm sub}}})$ to obtain good performance. This is a new advantage for the sub-connected structure in reducing the feedback overhead.

\section{Conclusions}
In this paper, we proposed the NUQ-HYP-Full scheme and the NUQ-HYO-Sub scheme for the full-connected and the sub-connected structures in millimeter wave MIMO systems, respectively. The key idea of the proposed schemes is that the quantization bits are non-uniformly mapped to different coverage angles, according to the sparseness property of the millimeter wave in the angular domain. Both of the proposed schemes achieve at least $12.5\%$ feedback overhead reduction for a system with 144/36 transmitting/receiving antennas. Simulation results demonstrated that the proposed NUQ-HYP-Full scheme for the full-connected structure exhibits similar spectral efficiency as the fully digital precoding scheme and outperforms the UQ-OMP scheme. 
Simulation results also showed that the proposed NUQ-HYP-Sub scheme for the sub-connected structure achieves similar spectral efficiency as the optimal unconstrained precoding scheme. Furthermore, we also observed that, the required number of quantization bits in the sub-connected structure to obtain near-optimal spectral efficiency was smaller than that in the full-connected structure, which provided a new insight to study low feedback overhead hybrid precoding schemes in millimeter wave MIMO systems.
Our future work will focus on the wideband and time-varying millimeter wave channel scenarios, where the delay and doppler shift become key points to design hybrid precoding matrices.

\end{document}